\documentclass[11pt, onecolumn, reqno]{amsart}

\usepackage[a4paper, top=1.5in, bottom=1.4in, left=.9in, right=.9in]{geometry}
\usepackage[toc,page]{appendix}
\usepackage{amssymb,amsmath,amsthm,amsfonts,graphicx,microtype,siunitx,booktabs,cite,tikz,makecell,dsfont,pgfplots,algorithm, algorithmic}
\usepackage{multirow,setspace}

\usetikzlibrary{decorations.pathreplacing,automata,calc,positioning,arrows}

\newtheorem{thm}{Theorem}
\newtheorem{prop}[thm]{Proposition}
\newtheorem{lem}[thm]{Lemma}

\newtheoremstyle{newstyle}
{2pt} 
{1pt} 
{\upshape} 
{0pt} 
{\bfseries} 
{.} 
{5pt} 
{} 

\theoremstyle{newstyle}
\newtheorem{definition}{Definition}

\newtheorem{example}{Example}

\usepackage{hyperref}
\hypersetup{
    backref =       true,
    pagebackref  =  true,
    colorlinks =    true,
    linkcolor =     [rgb]{0.0,0.0,0.8},
    anchorcolor =   [rgb]{0.0,0.0,0.8},
    citecolor =     green,
    filecolor =     [rgb]{0.0,0.0,0.8},
    urlcolor =      [rgb]{0.0,0.0,0.8},
    pdftitle=       {Title},
    pdfsubject=     {Title},
    pdfauthor=      {A. Uthor}
}

\newlength
\figureheight 
\newlength
\figurewidth

\newcommand{\bigp}[2]{\left(}

\newcommand{\D}{D_{v,\textup{CSMA}}}   
\newcommand{\T}{T_{v,\textup{CSMA}}}   
\newcommand{\Pidle}{P^{\textup{IDLE}}(\pmb{\rho}^{(N)})}   
\newcommand{\Psucc}{P_v^{\textup{SUCC}}(\pmb{\rho}^{(N)})}   
	\newcommand{\Pcol}{P_v^{\textup{COL}}(\pmb{\rho}^{(N)})}   
\newcommand{\SNR}{{\sf SNR}}

\newcommand{\LineIf}[2]{ \STATE \algorithmicif\ {#1}\ \algorithmicthen\ {#2} \algorithmicend\ \algorithmicif }

\newcommand{\cadd}[1]{#1}
\newcommand{\add}[1]{#1}
\newcommand{\del}[1]{}

\title{\textbf{Performance Analysis of CSMA with Multi-Packet Reception: \\The Inhomogeneous Case}}

\author{Shwan Ashrafi$^{\ast}\quad$ Chen Feng$^{\dagger}\quad$ Sumit Roy$^{\ast}$
}
\thanks{$^{\ast}$ S.~Ashrafi and S.~Roy are with the EE Department, University of Washington, Seattle, WA. Emails: \texttt{shwan@uw.edu}, \texttt{roy@ee.washington.edu}.}
\thanks{$^{\dagger}$ C. Feng
is with the School of Engineering, University of British Columbia, Kelowna, Canada. Email: \texttt{chen.feng@ubc.ca}.}

\begin{document}

\maketitle
\thispagestyle{empty}

\section*{Abstract}

\setlength{\baselineskip}{8mm}
\singlespacing
The problem of Carrier Sense Multiple Access (CSMA) with multi-packet reception (MPR) is studied. Most prior work has focused on the homogeneous case, where all the mobile users are assumed to have identical packet arrival rates and transmission probabilities. The inhomogeneous case remains largely open in the literature. In this work, we make a first step towards this open problem by deriving throughput and delay expressions for inhomogeneous CSMA, with a particular focus on a family of MPR models called the ``all-or-nothing" symmetric MPR. This family of MPR models allows us to overcome several technical challenges associated with conventional analysis and to derive accurate throughput and delay expressions in the large-systems regime. Interestingly, this family of MPR models is still general enough to include a number of useful MPR techniques---such as successive interference cancellation (SIC), compute-and-forward (C\&F), and successive compute-and-forward (SCF)---as special cases. Based on these throughput and delay expressions, we provide theoretical guidelines for meeting quality-of-service requirements and for achieving global stability; we also evaluate the performances of various MPR techniques, highlighting the clear advantages offered by SCF.

\noindent

\textbf{Keywords:}
Performance Analysis, Mean Field Approximation, Inhomogeneous CSMA, Multi-packet Reception, Compute-and-Forward

\setlength{\hsize}{6.5in}

\setlength{\baselineskip}{7.1mm}

\section{Introduction}
Densification of wireless local area networks (WLANs) is a common response to the exponential increase in data traffic density. Co-channel interference resulting from simultaneous packet transmissions \cite{Cheng06jigsaw:solving} constitutes the fundamental limit to performance in such networks. Rather than avoiding such simultaneous transmissions, there has been increasing interest in proposing new random-access protocols that exploit such multiple access \cite{Huang:2008, Zheng:mpr:wireless:2006, Chan::2013, Chan:crosslayer:2004, Tan::2009, Zhang:tput:mpr:2010}, i.e., make use of advanced signal-processing techniques in the physical layer to allow multiple packet reception (MPR) \cite{Ghez::1988}, thereby improving the overall network performance.

Despite a large body of work on random-access protocols with MPR capability
\cite{Ghez::1988, Naware::2005, Chan::2004, Chan::2013, Jin::stability::2014, Wu::2014::CSMAperformance, Gau::nonpersistent::2009, GaiGK::2011, Chan:crosslayer:2004, Zhang:tput:mpr:2010, Bae::maxtput::2014, Tong:mag:2001},
some of their fundamental properties are still not well understood. For instance, most performance analysis of MPR-capable CSMA
focuses on the {\em homogeneous} case, where all the mobile users are assumed to have identical packet arrival rates and transmission probabilities \cite{Chan::2004, Chan::2013, Jin::stability::2014, GaiGK::2011, Chan:crosslayer:2004}. In reality, mobile users often have different packet arrival rates, leading to the \emph{inhomogeneous} case. However, the throughput and delay performances of MPR-capable CSMA for the \emph{inhomogeneous} case remain largely unknown in the literature.

In this paper, we make a first step towards this open problem by deriving throughput and delay expressions for inhomogeneous CSMA, with a particular focus on a family of MPR channel models of \cite{Ghez::1988} that we call the ``all-or-nothing" symmetric MPR. The use of the ``all-or-nothing" symmetric MPR not only allows us to overcome several fundamental challenges associated with conventional analysis \cite{Naware::2005, Chan::2013}, but also leads to simple throughput and delay expressions in the large-systems regime that provide theoretical guidelines for practical network design. For example, we show that how these expressions can be used to choose the transmission probabilities in order to meet quality-of-service (QoS) requirements. We also show that how to guarantee global stability and avoid metastability\footnote{A system with multiple stable states is called metastable. We give more details on metastability in Section \ref{sec::metastability}.} via a cross-layer design.

Interestingly, the ``all-or-nothing" symmetric MPR model is still general enough to include a number of promising MPR techniques---such as successive interference cancellation (SIC), compute-and-forward (C\&F), and successive compute-and-forward (SCF)---as special cases. As a new MPR technique, C\&F enables a receiver to recover simultaneously transmitted packets via decoding linear equations \cite{Nazer_compnfwd_harness2011, ZG14, OEN14, WNPS10, NNW13, Feng::2013}. 
Compared to other MPR techniques, C\&F achieves close-to-optimal performance \cite{Nazer:2013:CFoptimality, ZhuG14b} under single-user decoding, making it particularly attractive for practical applications.
In this paper, we evaluate the performances of various MPR techniques in terms of throughput, packet delay and service delay through both analysis and simulation. Our results highlight the clear advantages of SCF-based CSMA over SIC-based CSMA and conventional CSMA.

The main contributions of this work are as follows.
\begin{itemize}
\item \add{We apply a mean field approximation to analyze the performance of inhomogeneous persistent CSMA with a family of MPR techniques. Based on this assumption the queue length of each user evolves independently from other queues in the large systems regime when $N$ tends to infinity .}

\item \add{Using the mean field approximation, we then distinguish three regions for arrival rates, namely stable, bistable and unstable regions, and provide an algorithm to obtain the non-empty probability of each queue and determine the state of the system.}

\item We then derive throughput and delay expressions for inhomogeneous CSMA with a family of MPR techniques, making a first step towards an open problem. 

\item Using these expressions, we provide theoretical guidelines to meet quality-of-service requirements and to achieve global stability.

\item Based on these expressions, we evaluate the performances of various MPR techniques in terms of throughput, packet delay and service delay.
\end{itemize}

\section{Related Work}

The study of random-access protocols with MPR capability dates back to late 80's.
In their seminal work \cite{Ghez::1988, Ghez::1989}, 
Ghez, Verd\'{u} and Schwartz introduced a so-called symmetric MPR channel model
for slotted ALOHA (one of the simplest random-access protocols) and characterized the stability
condition in the large-systems regime. Following \cite{Ghez::1988, Ghez::1989}, Sant and Sharma 
studied the finite-user regime by focusing on a special class of the symmetric MPR model \cite{sant:mpr:2000}.
An asymmetric MPR model was analyzed by Naware, Mergen and Tong \cite{Naware::2005},
with a particular focus on the two-user case and the homogeneous case (due to several technical challenges
explained in \cite{Naware::2005}). 
The stability condition for slotted ALOHA with MPR under the \emph{inhomogeneous} case
remains largely open in the literature. It is only very recently that some progress has been made 
in our previous work \cite{Shwan::ISIT2015} for a family of MPR models that contains SIC, C\&F, and SCF as special cases.

The study of MPR was extended from slotted ALOHA to CSMA systems, starting from the work of Chan, Berger and Tong \cite{Chan::2004} 
in 2004. In particular, the maximum stable throughput and the service delay were derived in \cite{Chan::2004, Chan:crosslayer:2004} for the symmetric MPR model of \cite{Ghez::1988} under the homogenous case. Other works along this direction include \cite{Jin::stability::2014, Wu::2014::CSMAperformance, GaiGK::2011, Bae::maxtput::2014, Zhang:mpr:2009, Gau::nonpersistent::2009, Babich:2010:mpr}. 
Recently, CSMA with MPR has received renewed interest from the research community, mainly
due to rapid advances in multiple-input multiple-output (MIMO) technology for wireless LAN. For example, Tan \emph{et al.} \cite{Tan::2009} developed a CSMA-type system  with chain-decoding MPR technique and demonstrated its clear advantage over conventional 802.11 through prototype implementation.    
Wu \emph{et al.} \cite{Wu::2014::CSMAperformance} evaluated the performance of such system design in terms of saturated throughput and service delay. Bae \emph{et al.} \cite{Bae::maxtput::2014} studied the optimal transmission probability to maximize the stable 
throughput for CSMA systems where the access point~(AP) can decode up to $M$ simultaneous transmissions, which is a special case of the symmetric MPR model.
The analysis in all these works is for the homogeneous case. By contrast, our work studies the inhomogeneous case.

As a promising new physical-layer technique, C\&F has received much attention recently.
Most prior work focused on its physical-layer performance (see, e.g., \cite{Shiqiang::multisource2013,ZhiChen::multisourceCF2014, ElSoussi::multisourceCF::twc::2014, LiliWei::multisourceCF::twc::2012, ZhiChen::multisourceCF::itv::2014} for multi-source multi-relay networks with C\&F) with a few exceptions \cite{Goseling:2013:ITA, GoselingRandAccess2013,Goseling:2014:sign,Goseling:2015:randaccessIT}. The work of Goseling \emph{et al.} pioneered the throughput analysis of random-access
with C\&F for the homogeneous case. By contrast, our work analyzes both the throughput and delay performance for the inhomogeneous case, which is of particular interest from a practical point of view.
\section{System Model} \label{Section::prelim}

\subsection{Network Model}

In this work we focus on a random-access network where, as illustrated in Fig.~\ref{fig::MACdiagram}, $N$ mobile users contend to transmit packets to an access point (AP). 
We assume that time is slotted, i.e., all packet transmissions are slot synchronous, and packets are of constant length requiring $\kappa$ time slots. Each user belongs to one of $V$ possible classes $\mathcal{V} = \{1,\cdots,V\}$, where all users in the same class have identical arrival rate and transmission probability. This $V$-class model captures user heterogeneity in terms of packet arrival rates and transmission probabilities. 

Each user is equipped with an infinite buffer for storing packets
in a FIFO manner. Packets arrive at the buffer of a class-$v$ user according to
a Bernoulli process with rate
$\lambda_v$. That is, at each time slot, a new packet arrives into
the buffer of a class-$v$ user with probability $\lambda_v$. The arrival
processes are assumed to be independent across users.

\subsection{Channel Model}
We use a standard block-fading multiple-access channel model \cite{Tse:2004:DMT} as illustrated in Fig.~\ref{fig::MACdiagram}, where each user is equipped with a single antenna and the AP is equipped with $K$ antennas. We assume block-level synchronization. Over a block length of $m$ symbols, 
when there are $L$ active users communicating to the AP,  the received signal at the $j$-th antenna at the AP can be written as
\[
\mathbf{y}_j = \sum_{\ell = 1}^L h_{j, \ell} \mathbf{x}_\ell + \mathbf{z}_j.
\]
Here, $\mathbf{x}_\ell \in \mathbb{C}^m$ is the transmitted signal (codeword) at the $\ell$th active user
subject to the average power constraint $\frac{1}{m} E \| \mathbf{x}_\ell \|^2 \le \, P$,
$h_{j, \ell}$ is the channel fading coefficient between the $\ell$th active user and $j$th antenna of the AP, and $\mathbf{z}_j$ represents the
additive white Gaussian noise vector with i.i.d. $\mathcal{CN}(0, 1)$ entries (where $\mathcal{CN}(0, a)$ denotes a complex
Gaussian random variable with independent zero-mean, variance $a/2$, Gaussian random variables as its real and imaginary parts).

In the matrix form, the AP observes a channel-output matrix matrix $\mathbf{Y} \in \mathbb{C}^{K \times m}$
\begin{align*}
\mathbf{Y} = \mathbf{H} \mathbf{X} + \mathbf{Z},
\end{align*}
where $\mathbf{H} \in \mathbb{C}^{K \times L}$ is the channel matrix whose entry in the $j$-th row and $\ell$-th column is $h_{j, \ell}$, $\mathbf{X} \in \mathbb{C}^{L \times m}$ is the channel-input matrix with $\mathbf{x}_\ell$ as its $\ell$th row, and $\mathbf{Z}$ is the 
noise matrix with $\mathbf{z}_j$ as its $j$th row. We assume that the channel-coefficient matrix $\mathbf{H}$ is known to the receiver but unknown to the transmitters. We statistically model $\mathbf{H}$ to be i.i.d. with $\mathcal{CN}(0, 1)$ entries, i.e., the richly scattered Rayleigh-fading environment\footnote{This assumption can be extended to other statistical models of $\mathbf{H}$.}.

\begin{figure}[t]
\centering
\makeatletter
\if@twocolumn
\scalebox{.45}{\definecolor{cffffff}{RGB}{255,255,255}

\begin{tikzpicture}[y=0.80pt, x=0.80pt, yscale=-1.000000, xscale=1.000000, inner sep=0pt, outer sep=0pt]
    \path[draw=black,line join=miter,line cap=butt,even odd rule,line width=0.800pt]
      (500.7731,270.1544) -- (468.1838,270.1544) -- (468.5256,235.3591) --
      (480.9786,218.5026) -- (455.6479,218.4579) -- (468.5103,235.3329);
  \path[draw=black,line join=miter,line cap=butt,even odd rule,line
    width=0.677pt,rounded corners=0.0000cm] (104.0513,77.4883) rectangle
    (205.3827,142.7965);
  \path[draw=black,line join=miter,line cap=butt,even odd rule,line width=0.800pt]
    (205.3571,109.3801) -- (237.9464,109.3801) -- (237.6046,74.5848) --
    (225.6567,57.6020) -- (250.4823,57.6836) -- (237.6199,74.5586);
  \path[fill=black,line join=miter,line cap=butt,line width=0.800pt]
    (125,114.9406) node[above right] {\Large{user $1$}};
  \path[draw=black,line join=miter,line cap=butt,even odd rule,line
    width=0.677pt,rounded corners=0.0000cm] (104.2794,186.2138) rectangle
    (205.6108,251.5220);
  \path[draw=black,line join=miter,line cap=butt,even odd rule,line width=0.800pt]
    (205.5852,218.1056) -- (238.1745,218.1056) -- (237.8327,183.3103) --
    (225.8848,166.3275) -- (250.7104,166.4091) -- (237.8480,183.2841);
  \path[fill=black,line join=miter,line cap=butt,line width=0.800pt]
    (125,223.6661) node[above right] {\Large{user $2$}};
  \path[draw=black,fill=black,line width=0.693pt] (152.5330,306.6171) ellipse
    (0.0589cm and 0.0518cm);
  \path[draw=black,fill=black,line width=0.693pt] (153.0381,331.6184) ellipse
    (0.0589cm and 0.0518cm);
  \path[draw=black,fill=black,line width=0.693pt] (153.5432,356.8722) ellipse
    (0.0589cm and 0.0518cm);
  \path[draw=black,line join=miter,line cap=butt,even odd rule,line
    width=0.677pt,rounded corners=0.0000cm] (107.3497,408.2626) rectangle
    (208.6811,473.5708);
  \path[draw=black,line join=miter,line cap=butt,even odd rule,line width=0.800pt]
    (208.6555,440.1544) -- (241.2448,440.1544) -- (240.9030,405.3591) --
    (228.9551,388.3763) -- (253.7807,388.4579) -- (240.9183,405.3329);
  \path[fill=black,line join=miter,line cap=butt,line width=0.800pt]
    (125,445.7149) node[above right]  {\Large{user $N$}};
  \path[xscale=-1.000,yscale=1.000,draw=black,line join=miter,line cap=butt,even
    odd rule,line width=0.677pt,rounded corners=0.0000cm] (-602.0789,258.2626)
    rectangle (-500.7475,323.5708);
  \path[fill=black,line join=miter,line cap=butt,line width=0.800pt]
    (522,290) node[above right] {\Large{Access} };
  \path[fill=black,line join=miter,line cap=butt,line width=0.800pt]
    (530,313) node[above right] {\Large{point} };
    
  \path[fill=black,line join=miter,line cap=butt,line width=0.800pt]
    (0.0000,0.0000) node[above right] (flowRoot4248) {$$};
    \path[draw=black,fill=cffffff,line join=miter,line cap=butt,even odd rule,line
      width=0.800pt] (461.2614,287.8275) -- (460.7753,287.8295) --
      (461.0225,253.7622) -- (473.4754,236.9057) -- (448.1447,236.8610) --
      (461.0072,253.7360);
    \path[draw=black,line join=miter,line cap=butt,even odd rule,line width=0.800pt]
      (460.6927,287.8476) -- (501.0357,287.8476);
  \path[draw=black,fill=cffffff,line join=miter,line cap=butt,even odd rule,line
    width=0.800pt] (452.4482,282.8071) -- (439.5858,265.9321) --
    (464.9165,265.9768) -- (452.4636,282.8333) -- (452.1217,317.6286) --
    (452.2110,317.6286) .. controls (501.8750,317.6286) and (500.5357,317.8509) ..
    (500.5357,317.8509);
  \path[draw=black,fill=cffffff,line width=1.033pt] (474.6429,296.4024) ellipse
    (0.0198cm and 0.0192cm);
  \path[<-, >=latex][draw=black,line join=miter,line cap=butt,miter limit=4.00,even odd
    rule,line width=1.280pt] (426.4167,195.0713) -- (263.2442,84.0813);
  \path[<-, >=latex][draw=black,line join=miter,line cap=butt,miter limit=4.00,even odd
    rule,line width=1.280pt] (416.3201,226.7308) -- (260.2187,187.4616);
  \path[<-, >=latex][draw=black,line join=miter,line cap=butt,miter limit=4.00,even odd
    rule,line width=1.280pt] (413.2896,285.3197) -- (406.2643,291.0035) --
    (268.2999,402.6241);
  \path[draw=black,fill=cffffff,line width=1.033pt] (469.1071,302.3176) ellipse
    (0.0198cm and 0.0192cm);
  \path[draw=black,fill=cffffff,line width=1.033pt] (462.9464,308.7908) ellipse
    (0.0198cm and 0.0192cm);

\end{tikzpicture}}
\else
\scalebox{.5}{\definecolor{cffffff}{RGB}{255,255,255}

\begin{tikzpicture}[y=0.80pt, x=0.80pt, yscale=-1.000000, xscale=1.000000, inner sep=0pt, outer sep=0pt]
    \path[draw=black,line join=miter,line cap=butt,even odd rule,line width=0.800pt]
      (500.7731,270.1544) -- (468.1838,270.1544) -- (468.5256,235.3591) --
      (480.9786,218.5026) -- (455.6479,218.4579) -- (468.5103,235.3329);
  \path[draw=black,line join=miter,line cap=butt,even odd rule,line
    width=0.677pt,rounded corners=0.0000cm] (104.0513,77.4883) rectangle
    (205.3827,142.7965);
  \path[draw=black,line join=miter,line cap=butt,even odd rule,line width=0.800pt]
    (205.3571,109.3801) -- (237.9464,109.3801) -- (237.6046,74.5848) --
    (225.6567,57.6020) -- (250.4823,57.6836) -- (237.6199,74.5586);
  \path[fill=black,line join=miter,line cap=butt,line width=0.800pt]
    (125,114.9406) node[above right] {\Large{user $1$}};
  \path[draw=black,line join=miter,line cap=butt,even odd rule,line
    width=0.677pt,rounded corners=0.0000cm] (104.2794,186.2138) rectangle
    (205.6108,251.5220);
  \path[draw=black,line join=miter,line cap=butt,even odd rule,line width=0.800pt]
    (205.5852,218.1056) -- (238.1745,218.1056) -- (237.8327,183.3103) --
    (225.8848,166.3275) -- (250.7104,166.4091) -- (237.8480,183.2841);
  \path[fill=black,line join=miter,line cap=butt,line width=0.800pt]
    (125,223.6661) node[above right] {\Large{user $2$}};
  \path[draw=black,fill=black,line width=0.693pt] (152.5330,306.6171) ellipse
    (0.0589cm and 0.0518cm);
  \path[draw=black,fill=black,line width=0.693pt] (153.0381,331.6184) ellipse
    (0.0589cm and 0.0518cm);
  \path[draw=black,fill=black,line width=0.693pt] (153.5432,356.8722) ellipse
    (0.0589cm and 0.0518cm);
  \path[draw=black,line join=miter,line cap=butt,even odd rule,line
    width=0.677pt,rounded corners=0.0000cm] (107.3497,408.2626) rectangle
    (208.6811,473.5708);
  \path[draw=black,line join=miter,line cap=butt,even odd rule,line width=0.800pt]
    (208.6555,440.1544) -- (241.2448,440.1544) -- (240.9030,405.3591) --
    (228.9551,388.3763) -- (253.7807,388.4579) -- (240.9183,405.3329);
  \path[fill=black,line join=miter,line cap=butt,line width=0.800pt]
    (125,445.7149) node[above right]  {\Large{user $N$}};
  \path[xscale=-1.000,yscale=1.000,draw=black,line join=miter,line cap=butt,even
    odd rule,line width=0.677pt,rounded corners=0.0000cm] (-602.0789,258.2626)
    rectangle (-500.7475,323.5708);
  \path[fill=black,line join=miter,line cap=butt,line width=0.800pt]
    (522,290) node[above right] {\Large{Access} };
  \path[fill=black,line join=miter,line cap=butt,line width=0.800pt]
    (530,313) node[above right] {\Large{point} };
    
  \path[fill=black,line join=miter,line cap=butt,line width=0.800pt]
    (0.0000,0.0000) node[above right] (flowRoot4248) {$$};
    \path[draw=black,fill=cffffff,line join=miter,line cap=butt,even odd rule,line
      width=0.800pt] (461.2614,287.8275) -- (460.7753,287.8295) --
      (461.0225,253.7622) -- (473.4754,236.9057) -- (448.1447,236.8610) --
      (461.0072,253.7360);
    \path[draw=black,line join=miter,line cap=butt,even odd rule,line width=0.800pt]
      (460.6927,287.8476) -- (501.0357,287.8476);
  \path[draw=black,fill=cffffff,line join=miter,line cap=butt,even odd rule,line
    width=0.800pt] (452.4482,282.8071) -- (439.5858,265.9321) --
    (464.9165,265.9768) -- (452.4636,282.8333) -- (452.1217,317.6286) --
    (452.2110,317.6286) .. controls (501.8750,317.6286) and (500.5357,317.8509) ..
    (500.5357,317.8509);
  \path[draw=black,fill=cffffff,line width=1.033pt] (474.6429,296.4024) ellipse
    (0.0198cm and 0.0192cm);
  \path[<-, >=latex][draw=black,line join=miter,line cap=butt,miter limit=4.00,even odd
    rule,line width=1.280pt] (426.4167,195.0713) -- (263.2442,84.0813);
  \path[<-, >=latex][draw=black,line join=miter,line cap=butt,miter limit=4.00,even odd
    rule,line width=1.280pt] (416.3201,226.7308) -- (260.2187,187.4616);
  \path[<-, >=latex][draw=black,line join=miter,line cap=butt,miter limit=4.00,even odd
    rule,line width=1.280pt] (413.2896,285.3197) -- (406.2643,291.0035) --
    (268.2999,402.6241);
  \path[draw=black,fill=cffffff,line width=1.033pt] (469.1071,302.3176) ellipse
    (0.0198cm and 0.0192cm);
  \path[draw=black,fill=cffffff,line width=1.033pt] (462.9464,308.7908) ellipse
    (0.0198cm and 0.0192cm);

\end{tikzpicture}}
\fi
\makeatother
\caption{A random access wireless channel where $N$ single-antenna users compete to gain access to a channel and transmit packets to a AP with $K$ antennas.}
\label{fig::MACdiagram}
\end{figure}
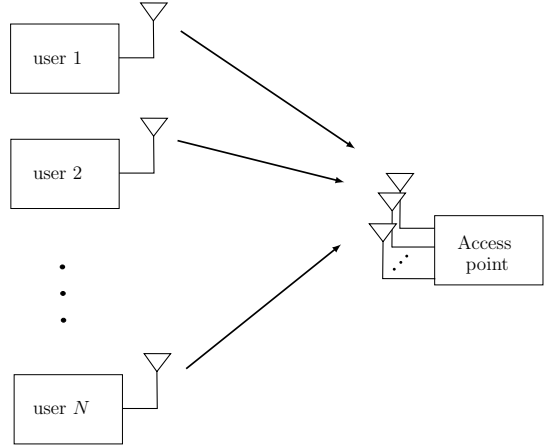

\subsection{Symmetric MPR Model} We now introduce the so-called symmetric MPR model defined in \cite{Ghez::1988} that we call symmetric MPR model.

\begin{definition}[Symmetric MPR]
Given that $L$ packets are being transmitted in a block time, the probability of successful recovery of exactly \add{$k$} packets is defined as

\makeatletter
\if@twocolumn
\begin{align*}
 q_{k,L} &= \Pr ( k \text{ packets are decoded successfully } | \\
 &\qquad\qquad L \text{ packets are transmitted}).
\end{align*}
\else
\[
q_{k, L} = \Pr ( k \text{ packets are correctly received } | L \text{ packets are transmitted}).
\]
\fi
\makeatother

These $k$ packets are chosen, uniformly at random, from $L$ concurrent transmitters.
\end{definition}

As a special case, if $q_{k,L} = 0$ for all $k \in \{ 1, \ldots, L - 1 \}$, we have an ``all-or-nothing" symmetric MPR model in which the receiver either recovers every simultaneously transmitted packets or none of them.

\begin{definition}[All-or-Nothing Symmetric MPR]
Given that $L$ packets are being transmitted in a block time, the probability of successful recovery of these $L$ packets is defined as
\makeatletter
\if@twocolumn
\begin{align*}
q_L &= \Pr ( L \text{ packets are decoded successfully } | \\
 &\qquad\qquad L \text{ packets are transmitted}).
\end{align*}
\else
\[
q_L = \Pr ( L \text{ packets are correctly received } | L \text{ packets are transmitted}).
\]
\fi
\makeatother
\end{definition}

Clearly, the ``all-or-nothing" symmetric MPR model generalizes the classical collision channel model, which has $q_1 = 1$ and $q_L = 0$ for all $L > 1$.
It also contains the $M$-user MPR model used in \cite{Guo:2009:kMPR} as a special case, 
which assumes that $q_L = 1$ for all $L \le M$ and
$q_M = 0$ for all $L > M$, where $M$ is a threshold determined by practical constraints (such as the speed of Walsh-Hadamard transform \cite{Patel:sic:jsac:1994}). 

In this paper, we will mainly focus on the ``all-or-nothing" symmetric MPR model. The extension to the general symmetric MPR model will be  discussed in Section~\ref{sec:general:sym:mpr}.

\subsection{Inhomogeneous Persistent CSMA with Symmetric MPR} We are now ready to introduce an inhomogeneous persistent CSMA protocol under the symmetric MPR model. 
For simplicity, we assume that nodes can hear each other, i.e., each node is within the transmission range of any other node, so that the hidden node problem is avoided. 

The protocol works as follows. Each user $i$ senses the channel continuously until it finds the channel idle for a duration of time known as
the DIFS. Then, user $i$ (say, in class $v$) transmits a packet with probability $p_v$. If this transmission is successful (i.e., the AP is able to recover user $i$'s packet under 
the symmetric MPR model), user $i$ will receive an ACK from the AP after a short period of time known as the SIFS.

This protocol is very similar to $p$-persistent CSMA\footnote{In a $p$-persistent CSMA system, a user avoids transmission when it finds the channel busy; it then persistently monitors the channel and transmits a packet with probability $p$ as soon as the channel becomes idle.} except that the transmission probabilities are class dependent. For this reason, we call this protocol inhomogeneous persistent CSMA, or inhomogeneous CSMA for short. \add{Note that when the packet length is one time slot, persistent CSMA reduces to slotted ALOHA \cite[Chapter 4]{Bertsekas:1992:DN} and \cite{GaiGK::2011, TaoWanThesis:1999}. This is because carrier sensing will not help reducing collisions when packet duration is one slot.}

\section{Throughput and Delay Analysis 
for Inhomogeneous CSMA with MPR}\label{section::applications}

In this section, we apply the mean-field approximation to study the throughput and delay performance of inhomogeneous CSMA with MPR.
For ease of presentation, we will mostly focus on the ``all-or-nothing" symmetric MPR model. The extension to general symmetric MPR will be provided in 
Section~\ref{sec:general:sym:mpr}.

\subsection{Mean-Field Approximation}

Similar to conventional CSMA, the channel can be in either busy or idle state. As illustrated in Fig.~\ref{fig::CSMA_illustration}, a busy period is comprised of a transmission period followed by SIFS, DIFS durations and ACK packet transmission time, and an idle period is comprised of a DIFS duration. 
Without loss of generality, we assume that an idle period (i.e., a DIFS duration) is exactly one time slot and that a busy period consists of $\tau$ time slots (since the packets are of constant length requiring $\kappa$ time slots). We call a busy period or an idle period a \emph{super slot}. 
\add{Hence a super slot consists of either one single time slot when the channel is in idle state, or $\tau$ time slots which is equivalent to a transmission period followed by SIFS, DIFS periods and one ACK.}

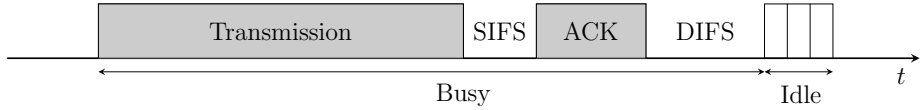
\begin{figure}[!t]
\centering
\makeatletter
\if@twocolumn
\scalebox{.58}{\begin{tikzpicture}[>=stealth,scale=1.5]
   \def\x{-1}
   \def\y{.6}
   \def\dx{4}
   \draw[thick,->] (-2,0) -- (8, 0);
   \draw[] (\x,0) -- (\x+\dx,0) -- (\x+\dx,\y) -- (\x,\y) -- (\x,0);
   \fill[gray!40!white, draw=black] (\x,0) rectangle(\x+\dx,\y);
   \node[] at (\x+2,.3) {\large Transmission};
   \node[] at (\x+4,-.4) {\large Busy};
   \node[] at (\x+\dx+.4,.3) {\large SIFS};
   \def\xx{\x+\dx+.8}
   \def\dxx{1.2}
   \draw (\xx,0) -- (\xx+\dxx,0) -- (\xx+\dxx,\y) -- (\xx,\y) -- (\xx,0);
   \fill[gray!40!white, draw=black] (\xx,0) rectangle (\xx+\dxx,\y);
   \node[] at (\xx+.6,.3) {\large ACK};
   \node[] at (\xx+1.85,.3) {\large DIFS};
   \def\z{\xx+\dxx+1.3}
   \def\dz{.25}
   \draw (\z,0) -- (\z+\dz,0) -- (\z+\dz,\y) -- (\z,\y) -- (\z,0);
   \draw (\z+\dz,\y) -- (\z+2*\dz,\y) -- (\z+2*\dz,0);
   \draw (\z+2*\dz,\y) -- (\z+3*\dz,\y) -- (\z+3*\dz,0);
   \node[] at (\z+.4,-.4) {\large Idle};
   
   \draw[<->] (\x,-.15) -- (\z,-.15);
   \draw[<->] (\z,-.15) -- (\z+3*\dz,-.15);
   
   \node[] at (\z+1.5,-.2) {\large $t$};
\end{tikzpicture}}
\else
\scalebox{.8}{\begin{tikzpicture}[>=stealth,scale=1.5]
   \def\x{-1}
   \def\y{.6}
   \def\dx{4}
   \draw[thick,->] (-2,0) -- (8, 0);
   \draw[] (\x,0) -- (\x+\dx,0) -- (\x+\dx,\y) -- (\x,\y) -- (\x,0);
   \fill[gray!40!white, draw=black] (\x,0) rectangle(\x+\dx,\y);
   \node[] at (\x+2,.3) {\large Transmission};
   \node[] at (\x+4,-.4) {\large Busy};
   \node[] at (\x+\dx+.4,.3) {\large SIFS};
   \def\xx{\x+\dx+.8}
   \def\dxx{1.2}
   \draw (\xx,0) -- (\xx+\dxx,0) -- (\xx+\dxx,\y) -- (\xx,\y) -- (\xx,0);
   \fill[gray!40!white, draw=black] (\xx,0) rectangle (\xx+\dxx,\y);
   \node[] at (\xx+.6,.3) {\large ACK};
   \node[] at (\xx+1.85,.3) {\large DIFS};
   \def\z{\xx+\dxx+1.3}
   \def\dz{.25}
   \draw (\z,0) -- (\z+\dz,0) -- (\z+\dz,\y) -- (\z,\y) -- (\z,0);
   \draw (\z+\dz,\y) -- (\z+2*\dz,\y) -- (\z+2*\dz,0);
   \draw (\z+2*\dz,\y) -- (\z+3*\dz,\y) -- (\z+3*\dz,0);
   \node[] at (\z+.4,-.4) {\large Idle};
   
   \draw[<->] (\x,-.15) -- (\z,-.15);
   \draw[<->] (\z,-.15) -- (\z+3*\dz,-.15);
   
   \node[] at (\z+1.5,-.2) {\large $t$};
\end{tikzpicture}}
\fi
\makeatother
\caption{Illustration of the busy and idle periods (super slots) in an inhomogeneous persistent CSMA protocol.}
 \label{fig::CSMA_illustration}
\end{figure}

We are now ready to apply the mean-field approximation. Consider a sequence of systems \cadd{indexed} by the number of users $N$. For any fixed $N$, 
each class-$v$ user has the same packet arrival rate $\lambda_v^{(N)}$ and the same transmission probability $p_v^{(N)}$. When we scale the system, we \cadd{let} 
\[
\lambda_v^{(N)} N \to \tilde{\lambda}_v \mbox{ and } p_v^{(N)} N \to \tilde{p}_v, \mbox{ as } N \to \infty.
\]
That is, the arrival rate $\lambda_v^{(N)}$ scales like $\frac{\tilde{\lambda}_v}{N}$, and the transmission probability $p_v^{(N)}$ scales like $\frac{\tilde{p}_v}{N}$.
(Note that both $\tilde{\lambda}_v$ and $\tilde{p}_v$ can be larger than $1$, since they are no longer probabilities.)
Similarly, for any fixed $N$, we \cadd{denoted by $N_v^{(N)}$ the number of class-$v$ users and let} 
\[
\frac{N_v^{(N)}}{N} \to \beta_v, \mbox{ as } N \to \infty.
\]
That is, the fraction of class-$v$ users converges to $\beta_v$.

The mean-field approximation proceeds as follows. Assume that users' queues evolve independently of each other, and further assume that 
the queue-length evolution processes of the same class are statistically identical. These assumptions can be validated
in the large-systems limit, which is beyond the scope of this paper and will be discussed in a companion paper. 
Under these assumptions, the queue evolution of a class-$v$ user is a discrete-time \cadd{renewal} process in which
each renewal cycle is a super slot.

Let $\rho_v^{(N)}$ be the utilization probability (i.e., the probability that the queue is non-empty in a limiting super time slot) for a class-$v$ user.
Clearly, the packet arrival rate for a class-$v$ user is $\lambda_v^{(N)}$. The average throughput (i.e., the average number of packets transmitted successfully in a time slot) of a class-$v$ user is given by 
\begin{equation}\label{eq:average_throughput}
R_v\left( \pmb{\rho}^{(N)} \right) = \frac{P_v\left( \pmb{\rho}^{(N)} \right)}{P^{\textsc{idle}} \left( \pmb{\rho}^{(N)} \right) + \tau \left (1 - P^{\textsc{idle}} \left( \pmb{\rho}^{(N)} \right) \right)}
\end{equation}
where  $\pmb{\rho}^{(N)}  \triangleq \left( \rho_1^{(N)}, \ldots, \rho_V^{(N)} \right)$, and
\begin{equation} \label{eq:idle:prob}
P^{\textsc{idle}} \left( \pmb{\rho}^{(N)} \right) \triangleq \prod_{u=1}^V \left(1-\rho_u^{(N)} p_u^{(N)} \right)^{N_u^{(N)}}
\end{equation}
is the probability that a super slot is idle, and
\makeatletter
\if@twocolumn 
\begin{align*}
&P_v\left( \pmb{\rho}^{(N)} \right) \triangleq \sum_{\substack{n_1 + \cdots + n_V \le M \\ \add{ 0\le n_u \le M\,\forall u\in \mathcal{V}} } }  \frac{n_v}{N_v^{(N)}} q_{n_1 + \cdots + n_V} \\
&\quad\cdot\prod_{u = 1}^V  {N_u^{(N)}  \choose n_u }  
  \left(\rho_u^{(N)} p_u^{(N)} \right)^{n_u }  \left(1-\rho_u^{(N)} p_u^{(N)} \right)^{N_u^{(N)}  - n_u }
\end{align*}
\else
\begin{equation}
P_v\left( \pmb{\rho}^{(N)} \right) \triangleq \sum_{\substack{n_1 + \cdots + n_V \le M \\ \add{ 0\le n_u \le M\,\forall u\in \mathcal{V}} } }  \frac{n_v}{N_v^{(N)}} q_{n_1 + \cdots + n_V} \prod_{u = 1}^V  {N_u^{(N)}  \choose n_u }  
  \left(\rho_u^{(N)} p_u^{(N)} \right)^{n_u }  \left(1-\rho_u^{(N)} p_u^{(N)} \right)^{N_u^{(N)}  - n_u }
\end{equation}
\fi
\makeatother
is the probability that a packet of a given class-$v$ user is successfully transmitted in a super slot.

To better understand Eq.~\eqref{eq:average_throughput}, one shall notice that the numerator $P_v\left( \pmb{\rho}^{(N)} \right)$ corresponds to the average number of packets transmitted successfully in a super slot for a class-$v$ user
and the denominator 
\[
P^{\textsc{idle}} \left( \pmb{\rho}^{(N)} \right) + \tau \left (1 - P^{\textsc{idle}} \left( \pmb{\rho}^{(N)} \right) \right) 
\] 
corresponds to the average length of a super slot. Therefore, Eq.~\eqref{eq:average_throughput} indeed gives the average throughput of a class-$v$ user under the
mean-field approximation. 

When $N \to \infty$, it is easy to verify that the (normalized) average throughput of a class-$v$ user admits the simple expression (\ref{limiting:rate:func:allornothingMPR}),
\makeatletter
\if@twocolumn 
\begin{floatEq}
\begin{align} \label{limiting:rate:func:allornothingMPR}
\lim_{N \to \infty} N R_v\left( \pmb{\rho}^{(N)} \right)  = \frac{ \rho_v^{(\infty)} \tilde{p}_v 
\left( q_1 + \frac{q_2}{1!} \gamma\left( \pmb{\rho}^{(\infty)} \right)  +
\cdots +  \frac{q_M}{(M - 1)!} \gamma^{M-1}\left( \pmb{\rho}^{(\infty)} \right) \right)  
e^{- \gamma \left( \pmb{\rho}^{(\infty)} \right)} }
{e^{- \gamma\left( \pmb{\rho}^{(\infty)} \right)}  + \tau\left(  1 -  e^{- \gamma\left( \pmb{\rho}^{(\infty)} \right)} \right) }
\end{align}
\end{floatEq}
\else
\begin{equation}  \label{limiting:rate:func:allornothingMPR}
\lim_{N \to \infty} N R_v\left( \pmb{\rho}^{(N)} \right)  = \frac{ \rho_v^{(\infty)} \tilde{p}_v 
\left( q_1 + \frac{q_2}{1!} \gamma\left( \pmb{\rho}^{(\infty)} \right)  +
\cdots +  \frac{q_M}{(M - 1)!} \gamma^{M-1}\left( \pmb{\rho}^{(\infty)} \right) \right)  
e^{- \gamma \left( \pmb{\rho}^{(\infty)} \right)} }
{e^{- \gamma\left( \pmb{\rho}^{(\infty)} \right)}  + \tau\left(  1 -  e^{- \gamma\left( \pmb{\rho}^{(\infty)} \right)} \right) }
\end{equation}
\fi
\makeatother
where
\begin{equation} \label{eq:gamma:func}
\gamma\left( \pmb{\rho}^{(\infty)} \right) \triangleq \sum_{u=1}^{V} \beta_u \tilde{p}_u \rho_u^{(\infty)}.
\end{equation}
To see why,
by (\ref{eq:idle:prob}) we have
\begin{align*}
P^{\textsc{idle}} \left( \pmb{\rho}^{(N)} \right) &= \prod_{u=1}^V \left(1-\rho_u^{(N)} p_u^{(N)} \right)^{N_u^{(N)}}\\
&= \prod_{u=1}^V \left(1-\frac{\rho_u^{(N)}  p_u^{(N)} N}{N} \right)^{\frac{N_u^{(N)}}{N} N}
\end{align*}
\begin{align*}
&\stackrel{(a)}{\to} \prod_{u=1}^V e^{- \rho_u^{(N)} \cdot (p_u^{(N)} N) \cdot \frac{N_u^{(N)}}{N}} \\
&\stackrel{(b)}{\to}  \prod_{u=1}^V e^{- \rho_u^{(\infty)}  \tilde{p}_u {\beta}_u} \\
&= \exp\left(  {- \underbrace{\sum_{u = 1}^V {\beta}_u \tilde{p}_u \rho_u^{(\infty)}}_{ \gamma\left( \pmb{\rho}^{(\infty)} \right) } }  \right) \\
&= e^{- \gamma\left( \pmb{\rho}^{(\infty)} \right)}
\end{align*}
as $N \to \infty$,
where step $(a)$ follows from the fact $\left( 1 + \frac{a}{N} \right)^{bN} \to e^{ab}$
and step $(b)$ makes use of the assumptions $p_u^{(N)} N \to \tilde{p}_u$ and $\frac{N_u^{(N)}}{N} \to \beta_u$.
This explains (\ref{eq:gamma:func}) and the denominator of (\ref{limiting:rate:func:allornothingMPR}).

Next, we will derive the numerator of (\ref{limiting:rate:func:allornothingMPR}). Note that
\makeatletter
\if@twocolumn 
\begin{align*}
&\qquad \prod_{u = 1}^V  {N_u^{(N)}  \choose n_u }  
  \left(\rho_u^{(N)} p_u^{(N)} \right)^{n_u }  \left(1-\rho_u^{(N)} p_u^{(N)} \right)^{N_u^{(N)}  - n_u} \\
&=\  \prod_{u = 1}^V \frac{N_u^{(N)} \cdots (N_u^{(N)} - n_u + 1)}{n_u!}   \left(\rho_u^{(N)} p_u^{(N)} \right)^{n_u }
\end{align*}
\begin{align*}
&\qquad\qquad\quad  \cdot\left(1-\rho_u^{(N)} p_u^{(N)} \right)^{N_u^{(N)}  - n_u} \\
&\stackrel{(c)}{\to} \prod_{u = 1}^V  \frac{1}{n_u!} 
  \left(\beta_u \tilde{p}_u \rho_u^{(\infty)}  \right)^{n_u } e^{- \beta_u \tilde{p}_u \rho_u^{(\infty)}}
  \end{align*}
  \begin{align*}
&= e^{- \gamma\left( \pmb{\rho}^{(\infty)} \right)} \prod_{u = 1}^V  \frac{1}{n_u!} 
  \left(\beta_u \tilde{p}_u \rho_u^{(\infty)}  \right)^{n_u }
\end{align*}
\else
\begin{align*}
&\qquad \prod_{u = 1}^V  {N_u^{(N)}  \choose n_u }  
  \left(\rho_u^{(N)} p_u^{(N)} \right)^{n_u }  \left(1-\rho_u^{(N)} p_u^{(N)} \right)^{N_u^{(N)}  - n_u} \\
&=\  \prod_{u = 1}^V \frac{N_u^{(N)} \cdots (N_u^{(N)} - n_u + 1)}{n_u!}   \left(\rho_u^{(N)} p_u^{(N)} \right)^{n_u }  \left(1-\rho_u^{(N)} p_u^{(N)} \right)^{N_u^{(N)}  - n_u} \\
&\stackrel{(c)}{\to} \prod_{u = 1}^V  \frac{1}{n_u!} 
  \left(\beta_u \tilde{p}_u \rho_u^{(\infty)}  \right)^{n_u } e^{- \beta_u \tilde{p}_u \rho_u^{(\infty)}} \\
&= e^{- \gamma\left( \pmb{\rho}^{(\infty)} \right)} \prod_{u = 1}^V  \frac{1}{n_u!} 
  \left(\beta_u \tilde{p}_u \rho_u^{(\infty)}  \right)^{n_u }
\end{align*}
\fi
\makeatother
as $N \to \infty$,
where step $(c)$ follows from the fact $(N_u^{(N)} - i)  \left(\rho_u^{(N)} p_u^{(N)} \right) \to 
\beta_u \tilde{p}_u \rho_u^{(\infty)}$ for all $i = 0, \ldots, n_u - 1$ and the fact
$\left(1-\rho_u^{(N)} p_u^{(N)} \right)^{N_u^{(N)}  - n_u}  \to  e^{- \beta_u \tilde{p}_u \rho_u^{(\infty)}}$.

Hence, 
\makeatletter
\if@twocolumn 
\begin{align*}
&N P_v\left( \pmb{\rho}^{(N)} \right) = N \sum_{\substack{n_1 + \cdots + n_V \le M \\ \add{ 0\le n_u \le M\,\forall u\in \mathcal{V}} } }  \frac{n_v}{N_v^{(N)}} q_{n_1 + \cdots + n_V} \\
&\prod_{u = 1}^V  {N_u^{(N)}  \choose n_u }  \left(\rho_u^{(N)} p_u^{(N)} \right)^{n_u }  \left(1-\rho_u^{(N)} p_u^{(N)} \right)^{N_u^{(N)}  - n_u } \\
&\to \!\!\!\!\!\!\! \sum_{\substack{n_1 + \cdots + n_V \le M \\ \add{ 0\le n_u \le M\,\forall u\in \mathcal{V}} } } \!\!\frac{n_v}{\beta_v} q_{n_1 + \cdots + n_V} 
e^{- \gamma\left( \pmb{\rho}^{(\infty)} \right)} \prod_{u = 1}^V 
 \frac{\left(\beta_u \tilde{p}_u \rho_u^{(\infty)}  \right)^{n_u }}{n_u!}  \\
&= \frac{e^{- \gamma\left( \pmb{\rho}^{(\infty)} \right)} }{\beta_v} \!\!\!\!\!\! \sum_{\substack{n_1 + \cdots + n_V \le M \\ \add{ 0\le n_u \le M\,\forall u\in \mathcal{V}} } } 
n_v q_{n_1 + \cdots + n_V}  \prod_{u = 1}^V 
\!\!\! \frac{\left(\beta_u \tilde{p}_u \rho_u^{(\infty)}  \right)^{n_u }}{n_u!}  \\
&\stackrel{(d)}{=} \!\!\frac{e^{- \gamma\left( \pmb{\rho}^{(\infty)} \right)} }{\beta_v}  
\beta_v \tilde{p}_v \rho_v^{(\infty)} \chi\left( \beta_1 \tilde{p}_1 \rho_1^{(\infty)} + \cdots + \beta_V \tilde{p}_V \rho_V^{(\infty)} \right)  \\
&=  \tilde{p}_v \rho_v^{(\infty)} e^{- \gamma\left( \pmb{\rho}^{(\infty)} \right)} \chi\left( \gamma\left( \pmb{\rho}^{(\infty)} \right) \right) \,,
\end{align*}
\else
\begin{align*}
N P_v\left( \pmb{\rho}^{(N)} \right) &= N \sum_{\substack{n_1 + \cdots + n_V \le M \\ \add{ 0\le n_u \le M\,\forall u\in \mathcal{V}} } }  \frac{n_v}{N_v^{(N)}} q_{n_1 + \cdots + n_V} \prod_{u = 1}^V  {N_u^{(N)}  \choose n_u }  \left(\rho_u^{(N)} p_u^{(N)} \right)^{n_u }  \left(1-\rho_u^{(N)} p_u^{(N)} \right)^{N_u^{(N)}  - n_u }
\end{align*}
\begin{align*}
&\to \sum_{\substack{n_1 + \cdots + n_V \le M \\ \add{ 0\le n_u \le M\,\forall u\in \mathcal{V}} } } \frac{n_v}{\beta_v} q_{n_1 + \cdots + n_V} 
e^{- \gamma\left( \pmb{\rho}^{(\infty)} \right)} \prod_{u = 1}^V 
 \frac{1}{n_u!} \left(\beta_u \tilde{p}_u \rho_u^{(\infty)}  \right)^{n_u }
 \end{align*}
 \begin{align*}
&= \frac{e^{- \gamma\left( \pmb{\rho}^{(\infty)} \right)} }{\beta_v}  \sum_{\substack{n_1 + \cdots + n_V \le M \\ \add{ 0\le n_u \le M\,\forall u\in \mathcal{V}} } } 
n_v q_{n_1 + \cdots + n_V}  \prod_{u = 1}^V 
 \frac{1}{n_u!} \left(\beta_u \tilde{p}_u \rho_u^{(\infty)}  \right)^{n_u }
 \end{align*}
 \begin{align*}
&\stackrel{(d)}{=} \frac{e^{- \gamma\left( \pmb{\rho}^{(\infty)} \right)} }{\beta_v}  
\beta_v \tilde{p}_v \rho_v^{(\infty)} \chi\left( \beta_1 \tilde{p}_1 \rho_1^{(\infty)} + \cdots + \beta_V \tilde{p}_V \rho_V^{(\infty)} \right)  \\
&=  \tilde{p}_v \rho_v^{(\infty)} e^{- \gamma\left( \pmb{\rho}^{(\infty)} \right)} \chi\left( \gamma\left( \pmb{\rho}^{(\infty)} \right) \right) \,,
\end{align*}
\fi
\makeatother
where step $(d)$ follows from the combinatorial identity 
\[
\!\!\!\!\sum_{\substack{n_1 + \cdots + n_V \le M \\ \add{ 0\le n_u \le M\,\forall u\in \mathcal{V}} } }\!\!\!\! n_v q_{n_1 + \cdots + n_V}  \prod_{u = 1}^V 
 \frac{1}{n_u!} x_u^{n_u }  = x_v \chi(x_1 + \cdots + x_V)
\]
with generating function $\chi(x) \triangleq q_1 + \frac{q_2}{1!} x + \frac{q_3}{2!} x^2 + \cdots + \frac{q_M}{(M-1)!}x^{M-1}$.
Then, the (normalized) average throughput of a class-$v$ user can be rewritten as
\begin{equation}
\lim_{N \to \infty} N R_v\left( \pmb{\rho}^{(N)} \right)  = \frac{ \rho_v^{(\infty)} \tilde{p}_v 
 \chi\left( \gamma  \left( \pmb{\rho}^{(\infty)} \right) \right)
e^{- \gamma \left( \pmb{\rho}^{(\infty)} \right)} }
{e^{- \gamma\left( \pmb{\rho}^{(\infty)} \right)}  + \tau\left(  1 -  e^{- \gamma\left( \pmb{\rho}^{(\infty)} \right)} \right) }\,.
\end{equation}
Finally, if the queue of a class-$v$ user is stable in the large-systems limit, then its arrival rate $\tilde{\lambda}_v / N$ is equal to its average throughput $\lim_{N \to \infty} R_v\left( \pmb{\rho}^{(N)} \right)$. Hence
\del{Similarly, if the utilization probability $\rho_v^{(\infty)} < 1$, the queue of a class-$v$ user in the large-systems limit 
is stable, and we have}
\begin{equation}
\tilde{\lambda}_v = \frac{ \rho_v^{(\infty)} \tilde{p}_v 
 \chi\left( \gamma  \left( \pmb{\rho}^{(\infty)} \right) \right)
e^{- \gamma \left( \pmb{\rho}^{(\infty)} \right)} }
{e^{- \gamma\left( \pmb{\rho}^{(\infty)} \right)}  + \tau\left(  1 -  e^{- \gamma\left( \pmb{\rho}^{(\infty)} \right)} \right) } \,.
\end{equation}

Recall that the duration of a time slot scales like $\frac{1}{N}$.
Thus, the limiting (infinite) system is a collection of $V$ coupled $M/M/1$ queues with arrival rate $\tilde{\lambda}_v$ and service rate 
\[
 \frac{ \tilde{p}_v 
 \chi\left( \gamma  \left( \pmb{\rho}^{(\infty)} \right) \right)
e^{- \gamma \left( \pmb{\rho}^{(\infty)} \right)} }
{e^{- \gamma\left( \pmb{\rho}^{(\infty)} \right)}  + \tau\left(  1 -  e^{- \gamma\left( \pmb{\rho}^{(\infty)} \right)} \right) }.
\]
The coupling is due to the limiting utilization probabilities $\pmb{\rho}^{(\infty)}$. This proves our main result of mean-field approximation in the following theorem.

\begin{thm}\label{thm:main}
Under the mean-field approximation, the limiting (infinite) system is a collection of $V$ coupled queues,
where the $v\,$th queue is $M/M/1$ with arrival rate $\tilde{\lambda}_v$ and service rate 
\[
 \mu_v\left( \gamma  \left( \pmb{\rho}^{(\infty)} \right) \right) \triangleq \frac{ \tilde{p}_v 
 \chi\left( \gamma  \left( \pmb{\rho}^{(\infty)} \right) \right)
e^{- \gamma \left( \pmb{\rho}^{(\infty)} \right)} }
{e^{- \gamma\left( \pmb{\rho}^{(\infty)} \right)}  + \tau\left(  1 -  e^{- \gamma\left( \pmb{\rho}^{(\infty)} \right)} \right) }.
\]
\end{thm}

\subsection{Stability of Limiting System}
Intuitively, if the system is stable the arrival rate is equal to the average throughput. This motivates our definition of global stability as follows. Consider the following system of equations
\begin{equation}\label{eq:system}
\forall v, \ \tilde{\lambda}_v = 
 \frac{ \rho_v \tilde{p}_v 
 \chi\left( \gamma  \left( \pmb{\rho}\right) \right)
e^{- \gamma \left( \pmb{\rho}\right)} }
{e^{- \gamma\left( \pmb{\rho}\right)}  + \tau\left(  1 -  e^{- \gamma\left( \pmb{\rho} \right)} \right) }.
\end{equation}
Let $\pmb{\rho}^* \in [0, 1)^V$ be one solution of (\ref{eq:system}). We say the limiting system is \emph{globally stable}\footnote{Note that this definition of global stability is unconventional, since it is based on the computational uniqueness instead of the stochastic behavior of the limiting system. Interestingly, this definition leads to the true stability region (based on the stochastic behavior), as shown in our companion paper.} with respect to $\tilde{\pmb{\lambda}}$ if $\pmb{\rho}^*$ is unique.

The \emph{limiting stability region} is defined as the set of vectors $\tilde{\pmb{\lambda}}$ such that the limiting system is globally stable for $\tilde{\pmb{\lambda}}$.
\begin{definition}\label{def::stability::limiting::system}
The limiting stability region is defined as

\makeatletter
\if@twocolumn 
\begin{align}
\nonumber
\Gamma &\triangleq \Bigg \{ { \tilde{\pmb{\lambda}}} \in \mathbb{R}_+^V : \exists\, \mbox{unique } \pmb{\rho} \in [0, 1]^V \mbox{ s.t. } \\
&\qquad\forall v,  \: \tilde{\lambda}_v =  \frac{ \rho_v \tilde{p}_v 
 \chi\left( \gamma  \left( \pmb{\rho}\right) \right)
e^{- \gamma \left( \pmb{\rho}\right)} }
{e^{- \gamma\left( \pmb{\rho}\right)}  + \tau\left(  1 -  e^{- \gamma\left( \pmb{\rho} \right)} \right) } \Bigg \},
\end{align}
\else
\begin{equation} 
\Gamma \triangleq \left \{ { \tilde{\pmb{\lambda}}} \in \mathbb{R}_+^V : \exists\, \mbox{unique } \pmb{\rho} \in [0, 1]^V \mbox{ s.t. } \forall v,  \: \tilde{\lambda}_v =  \frac{ \rho_v \tilde{p}_v 
 \chi\left( \gamma  \left( \pmb{\rho}\right) \right)
e^{- \gamma \left( \pmb{\rho}\right)} }
{e^{- \gamma\left( \pmb{\rho}\right)}  + \tau\left(  1 -  e^{- \gamma\left( \pmb{\rho} \right)} \right) } \right \} \,.
\end{equation}
\fi
\makeatother
\end{definition}

Although we define the limiting stability region based on the uniqueness of $\pmb{\rho}^*$, the resulting region turns out to be \emph{exact} in the sense that
it is precisely the true stability region of the limiting system. In other words, our mean-field approximation is asymptotically exact as $N \to \infty$.

We will provide an algorithm to determine the limiting stability region.
For simplicity, we assume that the function
\makeatletter
\if@twocolumn 
\begin{align}
&f(\gamma) \triangleq \frac{\gamma \chi(\gamma) e^{- \gamma}}{e^{-\gamma}+\tau(1-e^{-\gamma})} \\
&= \frac{\gamma \left( q_1 + \frac{q_2}{1!} \gamma + \frac{q_3}{2!} \gamma^2 + \cdots + \frac{q_M}{(M-1)!}\gamma^{M-1}
\right) e^{- \gamma}}{e^{-\gamma}+\tau(1-e^{-\gamma})}
\end{align}
\else
\begin{align}
f(\gamma) \triangleq \frac{\gamma \chi(\gamma) e^{- \gamma}}{e^{-\gamma}+\tau(1-e^{-\gamma})} = \frac{\gamma \left( q_1 + \frac{q_2}{1!} \gamma + \frac{q_3}{2!} \gamma^2 + \cdots + \frac{q_M}{(M-1)!}\gamma^{M-1}
\right) e^{- \gamma}}{e^{-\gamma}+\tau(1-e^{-\gamma})}
\end{align}
\fi
\makeatother
is unimodal. We notice 
that this assumption is rather mild. For example, this assumption holds as long as $M \le 2$. When $M > 2$, this assumption holds as long as $q_1 \le 2 q_2 \le \cdots \le M q_M$\del{, or $q_1 \ge 2 q_2 \ge \cdots \ge M q_M$}.
\add{From (\ref{eq:system}) we have
\begin{align} \nonumber
\lambda &\triangleq \sum_v \beta_v\tilde\lambda_v \\ \nonumber
&= \sum_v \beta_v  \frac{ \rho_v \tilde{p}_v 
 \chi\left( \gamma \right)
e^{- \gamma} }
{e^{- \gamma}  + \tau\left(  1 -  e^{- \gamma} \right) } \\
& = \frac{\gamma \chi(\gamma) e^{-\gamma}}{e^{- \gamma}  + \tau\left(  1 -  e^{- \gamma} \right) }= f(\gamma) \,. \label{def:lambda}
\end{align}
$\lambda$ is defined as the total input arrival rate. Let $\gamma^*$ be the maximizer of $f\,$. Define $\gamma_0 \triangleq \sum_v \beta_v\tilde{p}_v$ and let $\lambda_0 = f(\gamma_0)$. We then have the following lemma.}
\begin{lem} \label{lem::unimodal}
The function \cadd{$f(\gamma) = \frac{\gamma \chi(\gamma) e^{- \gamma}}{e^{-\gamma}+\tau(1-e^{-\gamma})}$} is unimodal if $q_1 \le 2 q_2 \le \cdots \le M q_M$. In particular, $f(\gamma)$ is unimodal if $M \le 2$.
\end{lem}

\begin{proof}
See Appendix \ref{appx::unimodal}.
\end{proof}


Algorithm~\ref{alg:rho} determines the set of arrival rates for which the system is globally stable. Intuitively all arrival rate vectors not in the stability region should constitute the instability region. However, it turns out that this intuition is not accurate and under a certain condition, there exits a third region of arrival rate vectors which make the system bistable. \cadd{The system is bistable if it is stable and the stationary distribution of queue length concentrates at two relatively separate regions.} Here we discuss this phenomenon in some details. Algorithm~\ref{alg:rho} not only gives a sufficient condition for the system to be bistabe, but determines whether a given arrival rate vector lies in either stable, bistable or unstable regions. The bistability can also be seen in CSMA systems without MPR capability \cite{Dai:csma:2013}.

Specifically, observe that depending on the input arrival rate, the equation $f(\gamma) = \lambda$ may have two, one or no solutions which corresponds to a bistable, stable or unstable system as illustrated in Fig.~\ref{fig:3regions}. When $\lambda > \lambda_0$ and $\gamma_0 > \gamma^*$ then the persistent CSMA system can potentially become bistable. To see why this is true assume a $2$-class network with $\tilde{p}_1=\tilde{p}_2$ and arrival rates $\tilde\lambda_1$ and $\tilde\lambda_2$ such that $\frac{\tilde\lambda_2}{\tilde\lambda_1}=r>1$ is fixed. As we increase $\tilde\lambda_1$ the two roots $\underline{\gamma}$ and $\overline{\gamma}$ of the equation $f(\gamma)=\lambda$ become smaller and larger respectively. Once the boundary of the stability region is reached $\pmb\rho = (\rho_1, 1)$. Further increase of $\tilde\lambda_2$ makes $\overline{\gamma}$ smaller such that $\pmb\rho_1 = (\rho_1, \rho_2)$ and $\rho_1 <1$, $\rho_2 < 1$. In other words there will be two stable points. 

 We now discuss how Algorithm~\ref{alg:rho} computes $\pmb{\rho}^{*}$ and determines which region a given rate lies.  For the rest of this discussion see Fig.~\ref{fig:unimodal}. The procedure presented in Algorithm \ref{alg:rho} takes a rate vector as input and determines $\pmb{\rho}^{*}$ and the state of the system, i.e., the region where the rate vector lies. First observe that $\gamma$ is a bounded quantity and $\gamma \le \gamma_0$. When the total input rate $\lambda < \lambda_0$ then the equation $\lambda = f(\gamma)$ has only one root as $f$ is a unimodal function and $\gamma \le \gamma_0$ should hold. When $\lambda \ge \lambda_0$ and $\gamma_0 \le \gamma^*$ then $\lambda = f(\gamma)$ has no root and the system is unstable. On the other hand when $\lambda \ge \lambda_0$ and $\gamma_0 > \gamma^*$, the equation $\lambda = f(\gamma)$ has two distinct roots and the system may be bistable or stable depending on whether $\pmb\rho_1 = \left( \lambda_1 / \mu_1(\underline{\gamma}), \cdots, \lambda_V / \mu_V(\underline{\gamma}) \right)$ and $\pmb\rho_2 = \left( \lambda_1 / \mu_1(\overline{\gamma}), \cdots, \lambda_V / \mu_V(\overline{\gamma}) \right)$ are both valid or either of them is valid.

\begin{figure}[t]
\begin{center}
\makeatletter
\if@twocolumn 
\setlength\figureheight{5.5cm} 
\setlength\figurewidth{0.83\columnwidth} 
\def\sx{.7}
\def\sy{.8}
\else
\setlength\figureheight{6.5cm} 
\setlength\figurewidth{.65\columnwidth} 
\def\sx{1}
\def\sy{1}
\fi
\makeatother
\scalebox{1}{\input{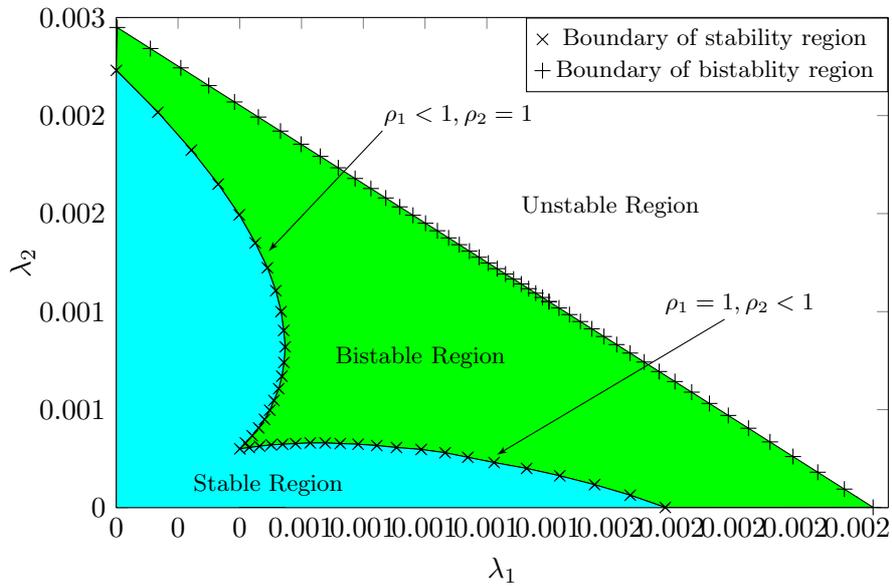} }
\caption{Illustration of stable, bistable and unstable regions.}
\label{fig:3regions}
\end{center}
\end{figure}

\begin{figure}[t]
\begin{center}
\makeatletter
\if@twocolumn 
\setlength\figureheight{5.5cm} 
\setlength\figurewidth{0.83\columnwidth} 
\def\sx{.7}
\def\sy{.85}
\else
\setlength\figureheight{6.5cm} 
\setlength\figurewidth{.65\columnwidth} 
\def\sx{1}
\def\sy{1}
\fi
\makeatother
\input{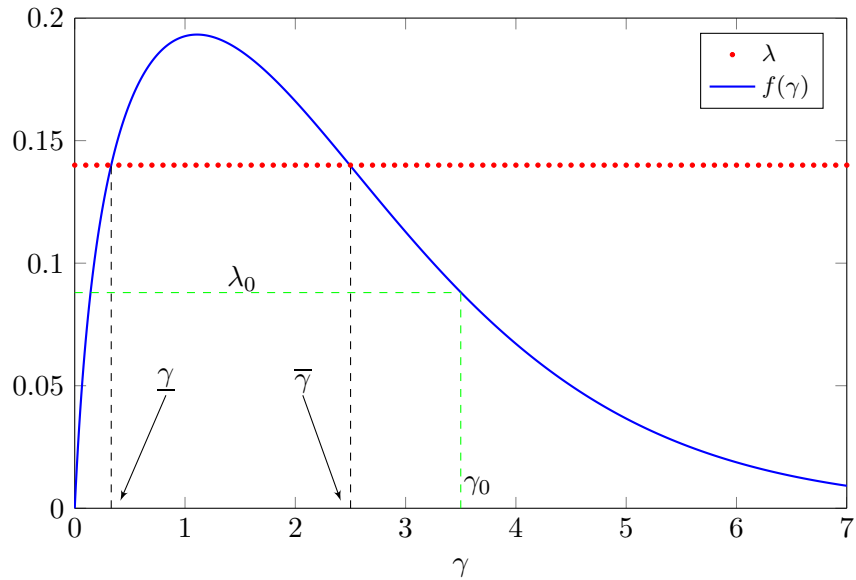} 
\caption{The unimodal function $f(\gamma)$ and its intersection with total input arrival rate.}
\label{fig:unimodal}
\end{center}
\end{figure}

\begin{algorithm}[tb]    

\caption{ A procedure to determine $\pmb{\rho}^{*}$ and the state of a $V$-class system.  }           
\label{alg:rho}                  

\begin{algorithmic}[1]              
\REQUIRE~~ $\pmb\lambda^V = (\lambda_1, \cdots, \lambda_V)$
 
\ENSURE ~~ $\pmb{\rho}^{*}$ and the state of system.

\IF{$\lambda < \lambda_0$}
   \STATE Find $\underline{\gamma}$ as the root of $f(\gamma) - \lambda = 0$
   \RETURN $\pmb{\rho}^{*} = \left(\lambda_1/\mu_1(\underline{\gamma}),\cdots, \lambda_V/\mu_V(\underline{\gamma})\right)$
\ELSE
\IF{$\lambda >\lambda_0$  \&  $\gamma_0 > \gamma^*$}
   \STATE Set $\underline{\gamma}$ to be the root of $f(\gamma) - \lambda=0$ in $[0, \gamma^*]$
   \STATE Set $\overline{\gamma}$ to be the root of $f(\gamma) - \lambda=0$ in $(\gamma^*, \gamma_0]$
   \STATE Set $\pmb\rho_1 = \left( \lambda_1 / \mu_1(\underline{\gamma}), \cdots, \lambda_V / \mu_V(\underline{\gamma}) \right)$
   \STATE Set $\pmb\rho_2 = \left( \lambda_1 / \mu_1(\overline{\gamma}), \cdots, \lambda_V / \mu_V(\overline{\gamma}) \right)$
   \IF{$\pmb\rho_1 < \pmb{1}\,\, \& \,\, \pmb\rho_2 < \pmb{1}$}
    \STATE state = \textsf{BISTABLE}
    \RETURN $\pmb{\rho}^{*}_1 = \pmb\rho_1, \pmb{\rho}^{*}_2=\pmb\rho_2$ and state
   \ELSE
   \LineIf {$\pmb\rho_1 < \pmb{1}$} {$\pmb{\rho}^{*} =  \pmb\rho_1$}
   \LineIf {$\pmb\rho_2 < \pmb{1}$} {$\pmb{\rho}^{*} =  \pmb\rho_2$}
    \STATE state = \textsf{STABLE}
    \RETURN $\pmb{\rho}^{*}$ and state
   \ENDIF
\ELSE
\STATE state = \textsf{UNSTABLE}
\RETURN state

\ENDIF
\ENDIF

\end{algorithmic}
\end{algorithm}
\subsection{Throughput and Delay Analysis}

Here, we characterize the throughput and delay performance based on the mean-field approximation, especially the limiting stability region. Our results apply to a network with an arbitrary number of classes.

\add{In this section we only focus on stable and bistable regions because if an arrival rate vector is in the unstable region then a rate control protocol like TCP Vegas can be used to reduce the arrival rate.}

Recall that, under the mean-field approximation, the average throughput of a class-$v$ user is given by $R_v\left( \pmb{\rho}^{(N)} \right)$, which is 
asymptotically exact. This leads to the following throughput result:

\begin{prop}\label{prop:throughput}
The aggregate throughput of the inhomogeneous persistent CSMA with ``all-or-nothing" symmetric MPR in a $V$-class network is given by
\makeatletter
\if@twocolumn
\begin{equation}\label{thm::tput}
\nonumber
R_{\textup{MPR-CSMA}}  = \sum_{v = 1}^V N_v^{(N)} R_v\left( \pmb{\rho}^{(N)} \right) + o(1),
\end{equation}
\else
\begin{equation}\label{thm::tput}
R_{\textup{MPR-CSMA}} = \sum_{v = 1}^V N_v^{(N)} R_v\left( \pmb{\rho}^{(N)} \right) + o(1)
\end{equation}
\fi
\makeatother
where the $o(1)$ term is understood as $N \to \infty$.
In particular, the saturated throughput is given by $\sum_{v} N_v R_v(\pmb{1})$, where $\pmb{1}$ is an all-one vector of length $V$.
\end{prop}

In order to apply Proposition~\ref{prop:throughput}, we need to obtain the value of the utilization probabilities $\pmb{\rho}^{(N)}$ for any given
$\{ \lambda_v \}$ and $\{ p_v \}$. Our strategy is to approximate $\pmb{\rho}^{(N)}$ with the limiting utilization probabilities $\pmb{\rho}^{(\infty)}$.
Specifically, we set $\tilde{\lambda}_v = \lambda_v N$ and $\tilde{p}_v = p_v N$ and then solve the system of equations given in \eqref{eq:system}\add{ to obtain a solution $\pmb{\rho}^{*}$. As we discussed previously, Algorithm \ref{alg:rho} can be used for this purpose.
}

Let us now look at the accuracy of  (\ref{thm::tput}) in a $2$-class network through simulations. Consider a network with $N$ users where half of the users belong to class $1$. We assume that $q_1= 0.96$, $q_2= 0.89$ and fix the transmission probability of class-$2$ users at $p_2 = 0.2\,$. We then compute the network aggregate throughput while varying the transmission probability of  class-$1$ users.

\begin{figure}[t]
\begin{center}
\makeatletter
\if@twocolumn 
\setlength\figureheight{5.5cm} 
\setlength\figurewidth{0.83\columnwidth} 
\else
\setlength\figureheight{6.5cm} 
\setlength\figurewidth{.65\columnwidth} 
\fi
\makeatother
\input{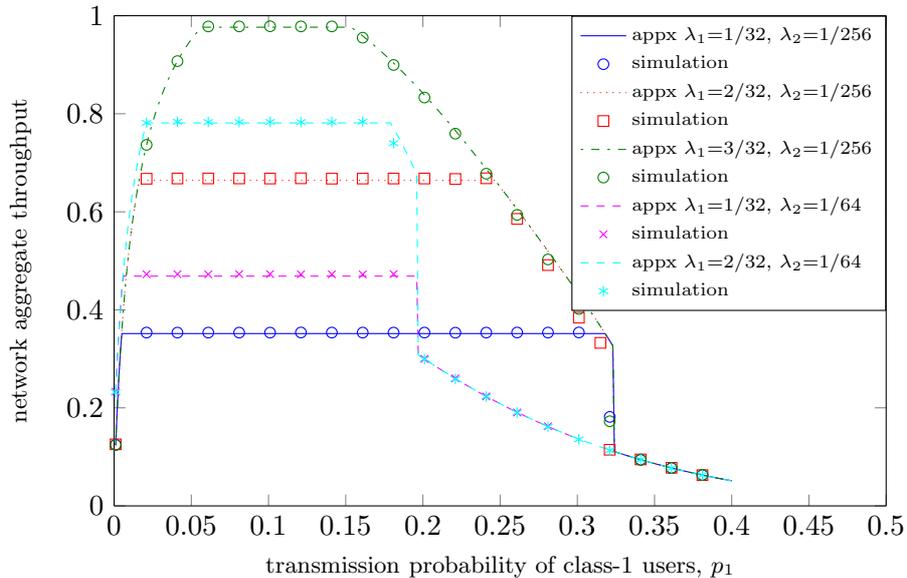} 
\caption{Aggregate throughput for a $2$-class network with $\kappa=10$, $p_2= 0.2$, $q_1= 0.96$, $q_2= 0.89$ and $N_1 = N_2 = 10$.}
\label{fig:aggregate::tput}
\end{center}
\end{figure}

Fig.~\ref{fig:aggregate::tput} shows the network aggregate throughput versus the transmission probability of class-$1$ users when $N = 20$. It can be seen that the aggregate throughput closely matches the analytical result for a variety of packet arrival rates.

Fig.~\ref{fig:max_sum_rate} shows the maximum aggregate throughput\add{, i.e., the maximum total system throughput while all the queues are stable,} versus the number of users $N$ when $p_1 = \frac{5}{6N}$ and $p_2 = \frac{7}{6N}$.
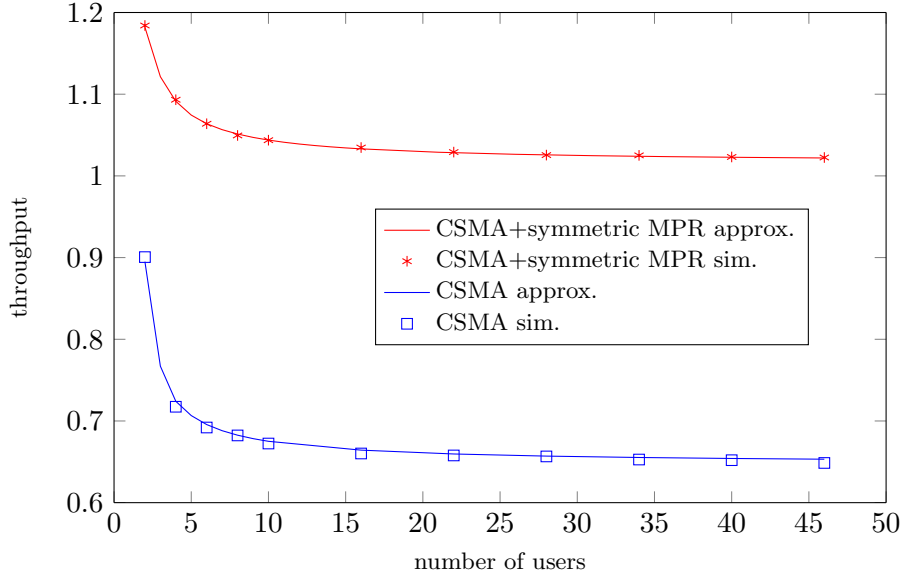
\begin{figure}[t]
\centering 
\makeatletter
\if@twocolumn 
\setlength\figureheight{5.5cm} 
\setlength\figurewidth{0.83\columnwidth} 
\else
\setlength\figureheight{6.5cm} 
\setlength\figurewidth{.65\columnwidth} 
\fi
\makeatother
%
%
\begin{tikzpicture}

\begin{axis}[%
width=0.95092\figurewidth,
height=\figureheight,
at={(0\figurewidth,0\figureheight)},
xlabel={\footnotesize number of users},
ylabel= {\footnotesize throughput},
legend style = { font=\footnotesize, at={(.9, .6)}, legend cell align=left},
scale only axis,
scaled ticks=false,
tick label style={/pgf/number format/fixed, /pgf/number format/precision=3},
separate axis lines,
every outer x axis line/.append style={black},
every x tick label/.append style={font=\color{black}},
xmin=0,
xmax=50,
every outer y axis line/.append style={black},
every y tick label/.append style={font=\color{black}},
ymin=0.6,
ymax=1.2,
]
\addplot [color=red, solid]
  table[row sep=crcr]{%
2	1.18196992347574\\
3	1.12143926207727\\
4	1.09130653479432\\
5	1.07444442366045\\
6	1.06380634230814\\
7	1.05659247834733\\
8	1.05126806182035\\
9	1.04714213759748\\
10	1.04385282295936\\
11	1.04119515406071\\
12	1.0390031328685\\
13	1.03716423665566\\
14	1.035599492861\\
15	1.03425184338621\\
16	1.03307903669521\\
21	1.02893836499179\\
26	1.0264268068267\\
31	1.02474094120594\\
36	1.02353557515696\\
41	1.02263499098517\\
46	1.02193262711377\\
};
\addlegendentry{CSMA+symmetric MPR approx.};

\addplot [color=red,only marks,mark=asterisk,mark options={solid}]
  table[row sep=crcr]{%
2	1.184075\\
4	1.09322\\
6	1.063823\\
8	1.0495\\
10	1.04371\\
16	1.0347\\
22	1.029133\\
28	1.02547\\
34	1.02511\\
40	1.02318\\
46	1.0225\\
};
\addlegendentry{CSMA+symmetric MPR sim.};

\addplot [color=blue,solid]
  table[row sep=crcr]{%
2	0.8950\\
3	0.7671\\
4	0.7241\\
5	0.7066\\
6	0.6956\\
7	0.6881\\
8	0.6826\\
9	0.6784\\
10	0.6751\\
16	0.6643\\
22	0.6596\\
28	0.6570\\
34	0.6553\\
40	0.6541\\
46	0.6532\\
};
\addlegendentry{CSMA approx.};

\addplot [color=blue,only marks,mark=square,mark options={solid}]
  table[row sep=crcr]{%
2	0.9006\\
4	0.7174\\
6	0.692\\
8	0.6823\\
10	0.6724\\
16	0.6602\\
22	0.6578\\
28	0.6566\\
34	0.6528\\
40	0.6520\\
46	0.6486\\
};
\addlegendentry{CSMA sim.};

\end{axis}
\end{tikzpicture}%
\caption{Maximum aggregate throughput of CSMA with symmetric MPR with $K=2$ and $\kappa=10$.}
\label{fig:max_sum_rate}
\end{figure}
First of all, the benefit of symmetric MPR based CSMA can be observed compared with conventional CSMA. Moreover, it can be seen that the aggregate throughput approximation matches the simulations considerably well.

We now turn our attention to study the delay performance. There are two types of delays of particular interest, namely, the
service delay and the total delay (packet delay). The \emph{service delay} of a packet is defined as the time it takes for 
this packet to be successfully decoded at the AP after it reaches the head of the queue, i.e., expected delivery time of a head of line (HOL) packet. The \emph{total delay}
of a packet is defined as the time it takes for the packet to be successfully decoded after it arrives in the queue.
Clearly, the total delay equals to the service delay plus the queuing delay.

Let $\D$ denote the service delay of a HOL packet of a class-$v$ user. Now consider a $V$-class network and assume that the system is stable. We have the following result  for the service delay.
\begin{prop}\label{prop:delay} The service delay of a class-$v$ user in the inhomogeneous persistent CSMA system with symmetric MPR is
\begin{align} \label{pCSMA::eq::service::delay}
\D = \frac{\rho_v^{(N)}}{\lambda_v} + o(1) \quad \text{for}\;\; v \in \mathcal{V}\,.
\end{align}
\end{prop}
\begin{proof} Suppose that $\rho^{(N)}_v \in (0, 1)$ for all $v \in \mathcal{V}$ is the unique solution to $\lambda_v = R_v(\pmb{\rho}^{(N)})$\,.  A class-$v$ user is then stable with the arrival rate $\lambda_v$ and we can apply Little's law to obtain
\[
\lambda_v \D = \rho_v^{(N)} + o(1)\,.
\]
This proves \eqref{pCSMA::eq::service::delay}.

We next provide an alternative proof for \eqref{pCSMA::eq::service::delay} which might give some additional insight.
Note that the success probability $\Psucc$ for a class-$v$ user can be obtained from the  throughput function as follows
\begin{align} \label{eq::Psucc}
\Psucc= \frac{P_v(\pmb{\rho}^{(N)})}{\rho_v^{(N)}}\,.   
\end{align}
Hence, the service delay for a class-$v$ user can be written in a recursive fashion as follows
\makeatletter
\if@twocolumn
\begin{align*}
\nonumber
&\D =   \tau \Psucc + (1 + \D ) \Pidle \\
     &\;\;+ (\tau + \D )\left(1 - \Pidle - \Psucc \right)\\
     &\;\; + o(1)\, .
\end{align*}
\else
\begin{align*} 
\nonumber
\D &=  \tau \Psucc + (1 + \D ) \Pidle \\
      &\qquad\qquad+ (\tau + \D )\left(1 - \Pidle - \Psucc \right) + o(1)\, .
\end{align*}
\fi
\makeatother
Therefore
\begin{align}\label{pCSMA::eq2::service::delay}
\D= \frac{\Pidle + \tau (1 - \Pidle )}{ \Psucc} + o(1) \,.
\end{align}
In fact the system is either in idle, collision or successful transmission states. If the system is in idle state, then the HOL packet transmission is delayed by one time slot. If the class-$v$ user is in the process of transmitting a packet successfully, the HOL packet transmission is delayed by $\tau$ time slots. Lastly, if the system is in collision state or there is any other successful transmission in progress, it takes $\tau$ more time slots for the HOL packet before it is successfully transmitted.
\end{proof}


Denote by $\T$ the total delay experienced by a packet in the inhomogeneous persistent CSMA system with symmetric MPR. We characterize the total delay in the following theorem.

\begin{thm} \label{thm::delay::CSMA}
Suppose that for all $v$, $\rho_v^{(N)} < 1$, i.e., the system is globally stable. The total delay of a class-$v$ packet in the inhomogeneous persistent CSMA system with symmetric MPR is
\makeatletter
\if@twocolumn
\begin{align} 
\nonumber
\T  &= \frac{\rho_v^{(N)} \Big (\frac{1}{\lambda_v}-\frac{1}{\tau} \Big) + \frac{\tau-1}{2}\left(1- \Pidle \right)}{1 - \rho_v^{(N)} } \\ \label{packetdelay::csma}
&\qquad\qquad\qquad\qquad\qquad\quad+ o(1)\,.
\end{align}
\else
\begin{align} \label{packetdelay::csma}
\T  = \frac{\rho_v^{(N)} \Big (\frac{1}{\lambda_v}-\frac{1}{\tau} \Big) + \frac{\tau-1}{2}\left(1- \Pidle \right)}{1 - \rho_v^{(N)} } + o(1)\,.
\end{align}
\fi
\makeatother
\end{thm}
\begin{proof}
See Appendix \ref{appendix::CSMA::delay}.
\end{proof}

\subsection{Adjusting Protocol Parameters}
The simple expressions derived before can be used to adjust the transmission probabilities to meet certain throughput or delay requirements. For instance, assume a two-class network with given arrival rates. From the stability region in \del{Definition}\add{Definition} \ref{def::stability::limiting::system}, we can immediately exclude the set of transmission probabilities for the two classes that destabilize the system or are too small that results is low throughput. Alternatively, one can exploit the delay expression in Theorem \ref{thm::delay::CSMA} to obtain the transmission probabilities such that a certain delay constraint is satisfied for a class. As a result, all the users in such a class will be stable. To further demonstrate the benefit of our delay analysis in system design, we give the following example
\begin{example} Consider a two-class network where $N_1$ users belong to class $1$ and $N_2$ users are in class $2$ with fixed arrival rates $\lambda_1$ and $\lambda_2$ respectively. Assume that a packet of a class-$v$ user is required to experience an average finite total delay no more than $T_v\,$. 

Since the delays are finite therefore all the queues are stable. Let $x_v = \rho_vp_v$ for $v\in \{1,2\}$. For the given $\lambda_v$'s, we can then compute $x_1^*$ and $x_2^*$ as the unique solution to $\lambda_v = R_v(x_1,x_2)\,$. In addition, from the delay constraint for each class and (\ref{packetdelay::csma}) it follows that 
\begin{align*}
p_v \geq \frac{\big( \frac{1}{\lambda_v}-\frac{1}{\kappa}  + T_v \big) x_v^*}{\frac{\kappa-1}{2}\left(1- P^{\textup{IDLE}}(x_1^*,x_2^*) \right)+ T_v} \quad \textup{for}\;\; v \in \{1,2\}\, .
\end{align*}
The above gives the design criteria for transmission probabilities of every class given the delay constraints. A similar approach can be applied to the general case of $V>2$ classes.

\end{example}
Lastly, if in a real scenario there are more than two classes of high-end and low-end users, as long as the requirement for the arrival rates of some users is such that they are below the arrival rates of high/low-end users, by a dominant system argument we see that these users enjoy a delay and throughput performance no worse than that experienced by high/low-end users. Obviously if new users with rate requirement higher than high-end users are required to join the network, we have to add another class of users, but the throughput and delay guarantee can be fulfilled easily based on the result we have obtained.

\subsection{Metastability \& Performance Guarantee} \label{sec::metastability}

Here, we explain how to avoid metastability based on \add{our results from the mean field approximation}. A system is called metastable if the stationary distribution of the underlying Markov chain is not unique.
As pointed out in \cite{Antunes:2006:metastability, Vvedenskaya::2007}, metastability is a highly undesirable property for a network. With metastability, the state of a network fluctuates -- over long periods of time -- between different stable states. Such long oscillations make it impossible to predict the average network performance in terms of throughput and delay. As a consequence, a proper level of quality-of-service cannot be guaranteed.

The following theorem provides a solution for inhomogeneous CSMA systems to avoid metastability. In particular, if the MPR technique is carefully designed, then the network is proven to be globally stable, and so metastability can be completely eliminated. This leads to an important design criterion, especially when quality-of-service guarantee is in great need.

\begin{thm} \label{thm::q::cond}
In the inhomogeneous persistent CSMA system with symmetric MPR model, metastability can be avoided if
\[
\sum_{v \in \mathcal{V}} N_v^{(N)} p_v < \overline{\gamma}
\]
and
\begin{align} \label{thm:cond1}
q_1 \leq 2q_2 \leq \cdots \leq Kq_K\,.
\end{align}
\end{thm}
\begin{proof}
The proof follows from Lemma \ref{lem::unimodal} and the fact that the arrival rates corresponding to bistable and unstable regions should be avoided, hence $\gamma_0 < \overline{\gamma}$.
\end{proof}
As an interesting consequence of Theorem \ref{thm::q::cond}, metastability is automatically avoided  when $K = 2$.
\add{
\subsection{Extention to General Symetric MPR Model} \label{sec:general:sym:mpr}
The symmetric MPR model we have discussed so far assumes that $q_{k, L} = 0$ for $k<L$. This assumption is only for the sake of a nicer presentation of the approximate stability region and to better understand the throughput and delay of persistent CSMA systems.
We relax this assumption in this section and briefly discuss a more general case of the symmetric MPR model where $q_{k, L} > 0$ for all $k\leq L$. Let
}
\makeatletter
\if@twocolumn
\begin{align}
\nonumber
p(n_1, \ldots, n_V; \pmb{\rho}^{(N)}) &\triangleq \prod_{v = 1}^V {N^{(N)}_v \choose  n_v} \left(\rho^{(N)}_v p^{(N)}_v \right)^{n_v} \\
&\! \left(1 - \rho^{(N)}_v p^{(N)}_v\right)^{N^{(N)}_v - n_v}.
\end{align}
\else
\begin{align}
p(n_1, \ldots, n_V; \pmb{\rho}^{(N)}) \triangleq \prod_{v = 1}^V {N^{(N)}_v \choose  n_v} \left(\rho^{(N)}_v p^{(N)}_v \right)^{n_v} \left(1 - \rho^{(N)}_v p^{(N)}_v\right)^{N^{(N)}_v - n_v}.
\end{align}
\fi
\makeatother
Then, the average throughput (i.e., the average number of packets transmitted successfully in a time slot) of a class-$v$ user is given by
\begin{align*} 
R_v(\pmb{\rho}^{(N)}) = \frac{P_v (\pmb{\rho}^{(N)})}{P^{\textsc{idle}} (\pmb{\rho}^{(N)}) + \tau \left (1 - P^{\textsc{idle}} (\pmb{\rho}^{(N)}) \right) }
\end{align*}
where
\begin{equation*} 
P_v(\pmb{\rho}^{(N)}) = \sum_{k=0}^M\sum_{\substack{n_1+ \cdots+ n_V = k \\ 0\le n_u \le k\,\forall u\in \mathcal{V}} } r_v(k) p(n_1, \ldots, n_V; \pmb{\rho}^{(N)})
\end{equation*}
and
\begin{equation*} 
r_v(k) = \frac{1}{N^{(N)}_v} \sum_{ \substack{0 \leq k_u \leq n_u \\ \forall u \in \mathcal{V} }} k_v q_{(k_1 + \cdots+ k_V),k }
\end{equation*}
is the average throughput of a user assuming that there are $k$ concurrent transmissions.

When $N \to \infty$, it is easy to verify that the (normalized) average throughput of a class-$v$ is given by (\ref{stab:limit:sym})
\makeatletter
\if@twocolumn
\begin{floatEq}
\begin{align} \label{stab:limit:sym}  
\lim_{N \to \infty} N R_v\left( \pmb{\rho}^{(N)} \right)  = \frac{e^{-\gamma_u\left( {\rho}^{(\infty)} \right)} \sum_{k=0}^N\sum_{n_1+\cdots+ n_V = k } \prod_{u=1}^V \frac{\gamma_u^{n_u}\left( {\rho}^{(\infty)} \right)}{n_u!} \bar r_v(k)}{e^{- \gamma_u\left( {\rho}^{(\infty)} \right)}  + \tau\left(  1 -  e^{- \gamma_u\left( {\rho}^{(\infty)} \right)} \right) }\, .
\end{align}
\end{floatEq}
\else
\begin{align} \label{stab:limit:sym}  
\lim_{N \to \infty} N R_v\left( \pmb{\rho}^{(N)} \right)  = \frac{e^{-\gamma_u\left( {\rho}^{(\infty)} \right)} \sum_{k=0}^N\sum_{n_1+\cdots+ n_V = k } \prod_{u=1}^V \frac{\gamma_u^{n_u}\left( {\rho}^{(\infty)} \right)}{n_u!} \bar r_v(k)}{e^{- \gamma_u\left( {\rho}^{(\infty)} \right)}  + \tau\left(  1 -  e^{- \gamma_u\left( {\rho}^{(\infty)} \right)} \right) }
\end{align}
\fi
\makeatother
where $\gamma_u\left( \rho^{(\infty)}_u\right) \triangleq \beta_u \rho^{(\infty)}_u  p_u$ and
\[
\bar r_v(k) = \frac{1}{\beta_v}\sum_{\substack{1 \leq k_u \leq n_u\\ \forall u \in \mathcal{V}}} k_v q_{(k_1+ \cdots+ k_V), k}\, .
\]
The general symmetric MPR model does not differ from the all-or-nothing MPR model in any fundamental way. The key difference is that the expressions for throughput and delay are more complicated in the general symmetric MPR model. The network aggregate throughput and total delay follow analogous to Propositions \ref{prop:throughput} and \ref{prop:delay}.
\section{symmetric MPR CSMA vs conventional CSMA} \label{simulations}

In this section, we first discuss C\&F, SCF, and SIC as special cases of MPR model. We then conduct extensive simulations to better demonstrate the benefit and accuracy of our approximations for throughput and delay of inhomogeneous CSMA with symmetric MPR. 
\subsection{Case Studies of Symmetric MPR Model} 
We explain how C\&F, SCF, and SIC can be used as MPR techniques. We give a brief summary of the C\&F and SCF schemes proposed in \cite{DBLP:journals/corr/NazerCNC15}, with a particular focus on the symmetric rates and complex-valued channel models. We refer our readers to \cite{Nazer:2013:CFoptimality, DBLP:journals/corr/NazerCNC15} for more details on C\&F.

\subsubsection{\bf Case Study 1: MPR via C\&F}

As our first case study, we explain how MPR can be achieved via C\&F technique.
Recall that the received signal at the AP is $\mathbf{Y} = \mathbf{H} \mathbf{X} + \mathbf{Z}$.
Rather than decoding $\mathbf{X}$ from $\mathbf{Y}$, the AP (which employs C\&F technique) will first decode
integer-linear combinations $\mathbf{A} \mathbf{X}$ and then invert these linear combinations to
recover the original signals $\mathbf{X}$. Here, $\mathbf{A} \in \mathbb{C}^{L \times L}$ is a target (invertible)
integer-valued matrix of which each row $\mathbf{a}_\ell$ corresponds to an integer combination.

Specifically, the AP chooses an equalizing filter matrix $\mathbf{B} \in \mathbb{C}^{L \times K}$ and computes
\begin{align*}
\mathbf{Y}_{\sf eff} &= \mathbf{B} \mathbf{Y} \\
&= \mathbf{A} \mathbf{X} + (\mathbf{B} \mathbf{H} - \mathbf{A}) \mathbf{X} + \mathbf{B} \mathbf{Z} \\
&= \mathbf{A} \mathbf{X}  + \mathbf{Z}_{\sf eff}
\end{align*}
where $\mathbf{Z}_{\sf eff} \triangleq (\mathbf{B} \mathbf{H} - \mathbf{A}) \mathbf{X} + \mathbf{B} \mathbf{Z}$
is the \emph{effective} noise matrix. In other words, C\&F transforms the original multiple access channel 
into $L$ point-to-point sub-channels
\[
\mathbf{y}_{{\sf eff}, \ell} = \mathbf{a}_\ell \mathbf{X} + \mathbf{z}_{{\sf eff}, \ell}, \ \ell = 1, \ldots, L
\]
where $\mathbf{y}_{{\sf eff}, \ell}$ and $\mathbf{z}_{{\sf eff}, \ell}$ are the $\ell$th rows of $\mathbf{Y}_{\sf eff}$
and $\mathbf{Z}_{\sf eff}$, respectively.

As shown in \cite{Nazer:2013:CFoptimality}, the optimal choice of the equalizing matrix $\mathbf{B}$ is given by
\[
\mathbf{B} = \mathbf{A} \mathbf{H}^T \left( \frac{1}{{\sf SNR}} \mathbf{I} + 
\mathbf{H} \mathbf{H}^T \right)^{-1}
\]
and an optimal choice of the integer matrix $\mathbf{A}$ can be found through lattice-reduction algorithms\footnote{We note that when the number of active users is small (less than 5), there exist some very efficient lattice-reduction algorithms. See, e.g., \cite{Nguyen:2009} for details.}.
The resulting achievable symmetric rate $R^{\sf CF}_{\sf sym}(\mathbf{H})$ is\footnote{All the $\log(\cdot)$ functions in this section are in base $2\,$.}
\begin{equation}\label{eq:rate_cf}
R^{\sf CF}_{\sf sym}(\mathbf{H}) = \min_{\ell = 1, \dots, L} \log\left( \frac{{\sf SNR}}{\sigma^2_{{\sf eff}, \ell} } \right)
\end{equation}
where $\sigma^2_{{\sf eff}, \ell}$ is the $\ell$th diagonal entry of the matrix ${\sf SNR}\ 
\mathbf{A}(\mathbf{I} + {\sf SNR}  \mathbf{H}^T \mathbf{H})^{-1} \mathbf{A}^T$.

For any given statistical model of $\mathbf{H}$, the success probability $q_L$ can be computed as
\[
q_L = \Pr \left(R < R^{\sf CF}_{\sf sym}(\mathbf{H})\right)
\]
where $R$ is the message rate of the transmitted packets.

\subsubsection{\bf Case Study 2: MPR via Successive C\&F}

Successive C\&F combines ideas from classical SIC and C\&F. Similar to C\&F, successive C\&F first recovers
a set of integer-linear combinations $\mathbf{A} \mathbf{X}$ and then find the original signals $\mathbf{X}$.
Rather than decoding each combination $\mathbf{a}_\ell \mathbf{X}$ in parallel, successive C\&F decodes these
combinations one at a time and makes use of already decoded combinations in subsequent decoding steps.

For any invertible integer matrix $\mathbf{A}$, the matrix $\mathbf{A}(\mathbf{I} + {\sf SNR} \mathbf{H}^T \mathbf{H})^{-1} \mathbf{A}^T$
admits a Cholesky decomposition
\[
\mathbf{A}(\mathbf{I} + {\sf SNR} \mathbf{H}^T \mathbf{H})^{-1} \mathbf{A}^T = \mathbf{L} \mathbf{L}^T
\]
where $\mathbf{L} \in \mathbb{C}^{L \times L}$ is a lower triangular matrix with strictly positive diagonal 
entries. As shown in \cite{Nazer:2013:CFoptimality}, the effective noise matrix $\mathbf{Z}_{\sf eff}$
can be written as
\[
\mathbf{Z}_{\sf eff} = \sqrt{{\sf SNR}}\,\, \mathbf{L} \mathbf{W}
\]
for some matrix $\mathbf{W}$ with unit generalized covariance matrix. Hence,
the resulting achievable symmetric rate $R^{\sf SCF}_{\sf sym}(\mathbf{H})$ is
\begin{equation}\label{eq:rate_scf}
R^{\sf SCF}_{\sf sym}(\mathbf{H}) = \min_{\ell = 1, \dots, L} \log\left( \frac{{\sf SNR}}{{\sf SNR} L_{\ell, \ell}^2 } \right)
\end{equation}
where $L_{\ell, \ell}$ is the $\ell$th diagonal entry of the matrix $\mathbf{L}$.
Similarly, the success probability for successive C\&F can be computed as
\[
q_L = \Pr(R < R^{\sf SCF}_{\sf sym}(\mathbf{H})).
\]

\subsubsection{\bf Case Study 3: MPR via SIC}

As pointed out in \cite{Nazer:2013:CFoptimality}, if the integer matrix $\mathbf{A}$ is chosen to be a permutation matrix, successive 
C\&F in the previous case study reduces to the well-known SIC (and each permutation matrix
corresponds to a corner point in the capacity region). Hence, successive C\&F, in general, offers higher achievable symmetric rates than SIC.
The success probability $q_L$ for SIC can be computed in a similar way as before, using the formula
\[
q_L = \Pr(R < R^{\sf SIC}_{\sf sym}(\mathbf{H})).
\]
Table~\ref{avg_dissm_tab} provides the values of $q_L$ for SIC, CF, SCF, assuming independent Rayleigh-fading environment. As a comparison, the performance of joint decoding (JD) is also provided,
where the symmetric rate $R^{\sf JD}_{\sf sym}(\mathbf{H})$ lies in the boundary of the capacity region
for any given channel matrix $\mathbf{H}$.

\begin{table}[!tbp]
\caption{Success probabilities $q_L$'s for SIC, CF, SCF, and JD under independent Rayleigh-fading environment and different message rates.}
\centering
\vspace{1mm}
\makeatletter
\if@twocolumn
   \def\w{.9}
\else
   \def\w{.6}
\fi
\makeatother
\resizebox{\w\columnwidth}{!}{
\begin{tabular}{c|c|c|c|cccc}
success prob. $q_L$ & \SNR~(dB) & message rate $R$ &$K$ & \small{SIC}&\small{C\&F }&{\small{\text{SCF}} } &\small{JD}\\
\hline\hline
 $q_1$ &\multirow{2}{*}{$6$} & \multirow{2}{5mm}{$\,\,1$} &  \multirow{2}{5mm}{$\,\,1$} & $0.78$ & $0.78$ &$0.78$ & $0.78$\\
 $q_2$ & & & & $0.46$ & $0.45$ &$0.57$ & $0.60$\\
\hline
 $q_1$ &\multirow{2}{*}{$15$} & \multirow{2}{5mm}{$\,\,2$} & \multirow{2}{5mm}{$\,\,1$} & $0.91$ & $0.91$ &$0.91$ & $0.91$\\
 $q_2$ & & & & $0.31$ & $0.61$ &$0.66$ & $0.80$\\
\hline
 $q_1$ &\multirow{3}{*}{$15$} & \multirow{3}{5mm}{$\,\,3$} &  \multirow{3}{5mm}{$\,\,2$} & $0.98$ & $0.98$ &$0.98$ & $0.98$\\
 $q_2$ & & & & $0.88$ & $0.92$ &$0.93$ & $0.95$\\
 $q_3$ & & & & $0.32$ & $0.70$ &$0.81$ & $0.91$\\
\end{tabular}
}
\label{avg_dissm_tab}	
\end{table}

As shown in Table~\ref{avg_dissm_tab}, $q_1$ remains the same for all the schemes under the three configurations. This is because when $L = 1$ (i.e., there is only one active user), the channel model reduces to a standard point-to-point channel. Also, it is observed that SCF significantly outperforms SIC in terms of $q_2$ and $q_3$, especially in the high-SNR regime. This is because SIC is a very special case of SCF. Moreover, the performance of SCF is close to that of JD even in the low-SNR regime.
\subsubsection{\bf Implementation Considerations.}
At first glance, the complexity of SCF appears to be significantly higher than that of SIC due to the use of lattice codes and lattice decoding. It turns out that SCF has essentially the same complexity as SIC for certain lattice codes constructed from convolutional codes and LDPC codes. Specifically, as shown in \cite[Appendix~G]{Feng::2013}, a slightly modified version of the Viterbi decoder can be used to implement a lattice decoder for lattice codes constructed through convolutional codes. In addition, as demonstrated in \cite{Pietro2012}, a family of high-performance lattice codes can be constructed via (non-binary) LDPC codes and decoded using an iterative message-passing algorithm whose complexity is essentially linear in the lattice dimension.

When a practical lattice code $\mathcal{C}$ described above is used, the previous symmetric rates \eqref{eq:rate_cf} and \eqref{eq:rate_scf} no longer hold, since they are information-theoretic bounds based on asymptotically \add{(as the lattice dimension goes to infinity)} good lattice codes. In this case, we can use the following formulas as suggested in \cite{Zamir:2014:LCbook}[Chapter~8] to estimate the achievable symmetric rates:
\makeatletter
\if@twocolumn
\begin{align}\label{eq:rate_cf_new}
\nonumber
R^{\sf CF}_{\sf sym}(\mathbf{H}) &= \min_{\ell = 1, \dots, L} \log\left( \frac{{\sf SNR}}{\sigma^2_{{\sf eff}, \ell} } \right) - \log\left( 2 \pi e G(\mathcal{L}') \right)\\
&\qquad\quad - \log\left( \frac{\mu(\mathcal{L}, P_e)}{2 \pi e} \right),
\end{align}
and
\begin{align}\label{eq:rate_scf_new}
\nonumber
R^{\sf SCF}_{\sf sym}(\mathbf{H}) &= \min_{\ell = 1, \dots, L} \log\left( \frac{{\sf SNR}}{{\sf SNR} L_{\ell, \ell}^2 } \right) \\
&\!\!\!\!\!\!\!\!\!- \log\left( 2 \pi e G(\mathcal{L}') \right) - \log\left( \frac{\mu(\mathcal{L}, P_e)}{2 \pi e} \right),
\end{align}
\else
\begin{equation}\label{eq:rate_cf_new}
R^{\sf CF}_{\sf sym}(\mathbf{H}) = \min_{\ell = 1, \dots, L} \log\left( \frac{{\sf SNR}}{\sigma^2_{{\sf eff}, \ell} } \right) - \log\left( 2 \pi e G(\mathcal{L}') \right) - \log\left( \frac{\mu(\mathcal{L}, P_e)}{2 \pi e} \right)
\end{equation}
and
\begin{equation}\label{eq:rate_scf_new}
R^{\sf SCF}_{\sf sym}(\mathbf{H}) = \min_{\ell = 1, \dots, L} \log\left( \frac{{\sf SNR}}{{\sf SNR} L_{\ell, \ell}^2 } \right) - \log\left( 2 \pi e G(\mathcal{L}') \right) - \log\left( \frac{\mu(\mathcal{L}, P_e)}{2 \pi e} \right)
\end{equation}
\fi
where $\mathcal{L}$ and $\mathcal{L}'$ are the coding lattice and shaping lattice for the lattice code $\mathcal{C}$, respectively. Here, $G(\mathcal{L}')$ is the normalized second moment for the shaping lattice $\mathcal{L}'$. (For example, $G(\mathcal{L}') = 1/12$ if the shaping lattice $\mathcal{L}'$ is a hypercube.) Also, $\mu(\mathcal{L}, P_e)$ is the normalized volume-to-noise ratio  (which is closely related to the coding gain) for a given error probability $P_e$ in nearest neighbor decoding of the lattice $\mathcal{L}$. Note that \eqref{eq:rate_cf_new} and \eqref{eq:rate_scf_new} apply to an arbitrary lattice code instead of asymptotically-good lattice codes, where $\log\left( 2 \pi e G(\mathcal{L}') \right)$ and $\log\left( \frac{\mu(\mathcal{L}, P_e)}{2 \pi e} \right)$ represent the shaping loss  and coding loss, respectively. We refer our readers to \cite[Chapter~3, ~8]{Zamir:2014:LCbook} for details.

Like many other advanced signal-processing techniques, SCF requires symbol-level synchronization and block-level synchronization. Such synchronizations can be achieved by leveraging some recent work on coherent transmission (see, e.g., \cite{Abari:2015:airshare}).  In addition, the requirement of block-level synchronization can be partly relaxed by using the technique developed in \cite{Wang:2015:asynch}. Finally, channel estimation is another important step of implementing SCF. We believe that it can be achieved through a combination of conventional channel estimation with blind C\&F technique proposed in \cite{Chen:2012:blindCF}, which is beyond the scope of this paper. For more details regarding the implementation aspects of physical layer network coding techniques see \cite{You:JSAC:2015, Mejri:WCNCW:2013, Cocco:asms:2012, Rossetto:spawc:2009}.

\subsection{Numerical Results} To validate our theoretical results and to further demonstrate the benefits of CSMA with SCF, we conduct simulations in various scenarios. Fig.~\ref{fig:saturated_tput} depicts saturated throughput, i.e.,  $R_v(\pmb{\pmb{1}})$ for CSMA with SCF, CSMA with SIC and conventional CSMA as a function of transmission probabilities \add{in a two-class network} when $N=5, 10$\add{, $N_1 = 3/5 N$} and \SNR~ $= 6\,$dB. It is evident that SCF-based CSMA performs close to CSMA with optimal joint decoding (JD) and significantly improves the saturated throughput compared to conventional CSMA. Further, CSMA with SCF universally outperforms CSMA with SIC offering high throughput for a wider range of transmission probabilities.
\begin{figure}[t]
\centering 
\makeatletter
\if@twocolumn 
\setlength\figureheight{5.5cm} 
\setlength\figurewidth{0.83\columnwidth} 
\else
\setlength\figureheight{6.5cm} 
\setlength\figurewidth{.65\columnwidth} 
\fi
\makeatother
\input{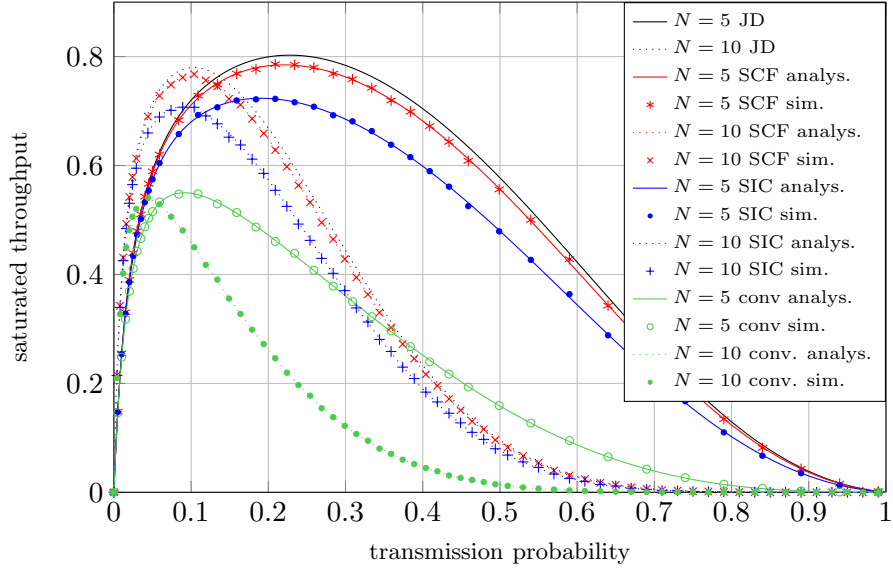} 
\caption{Saturated throughput benefit of SCF-based CSMA compared to CSMA with SIC and conventional CSMA in low \SNR~=~$6\,$dB, with \add{$p_2 = 0.8p_1$,} message rate $= 1$, $K=1$, $\kappa = 10$\,.}
\label{fig:saturated_tput}
\end{figure}

\begin{figure}[t]
\centering 
\makeatletter
\if@twocolumn 
\setlength\figureheight{5.5cm} 
\setlength\figurewidth{0.83\columnwidth} 
\else
\setlength\figureheight{6.5cm} 
\setlength\figurewidth{.65\columnwidth} 
\fi
\makeatother
\input{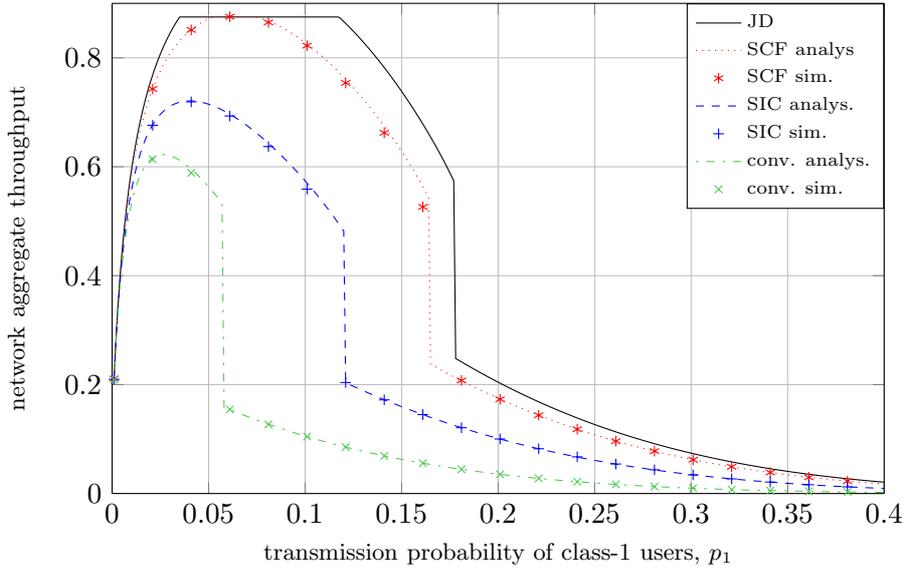} 
\caption{Throughput benefit of SCF-based CSMA compared to CSMA with SIC and conventional CSMA at \SNR~= $15\,$dB with message rate = $2$, $N_1=12$, $N_2=8$, $\lambda_1 = 1/16$, $\lambda_2 = 1/64$, \add{$p_2=1/4$}, $K=1$, $\kappa = 10$\,.}
\label{fig:MPRtput}
\end{figure}
Fig.~\ref{fig:MPRtput} shows the network aggregate throughput performance of CSMA with SCF, CSMA with SIC, conventional CSMA and CSMA with optimal JD for a network with $N_1=12$ class-1 users, $N_2=8$ class-2 users and \SNR~$=15\,$dB. CSMA with SCF clearly performs better than the other techniques and relatively close to optimal joint decoding. As the transmission probability of class-1 users $p_1$ varies from $0$ to $1$, eventually all the users become saturated (unstable queues) and the throughput of each user equals the rate function $R_v(\pmb{1})$ for all $v =\{1,2\}$.

\begin{figure}[t]
\centering 
\makeatletter
\if@twocolumn 
\setlength\figureheight{5.5cm} 
\setlength\figurewidth{0.83\columnwidth} 
\else
\setlength\figureheight{6.5cm} 
\setlength\figurewidth{.65\columnwidth} 
\fi
\makeatother
\input{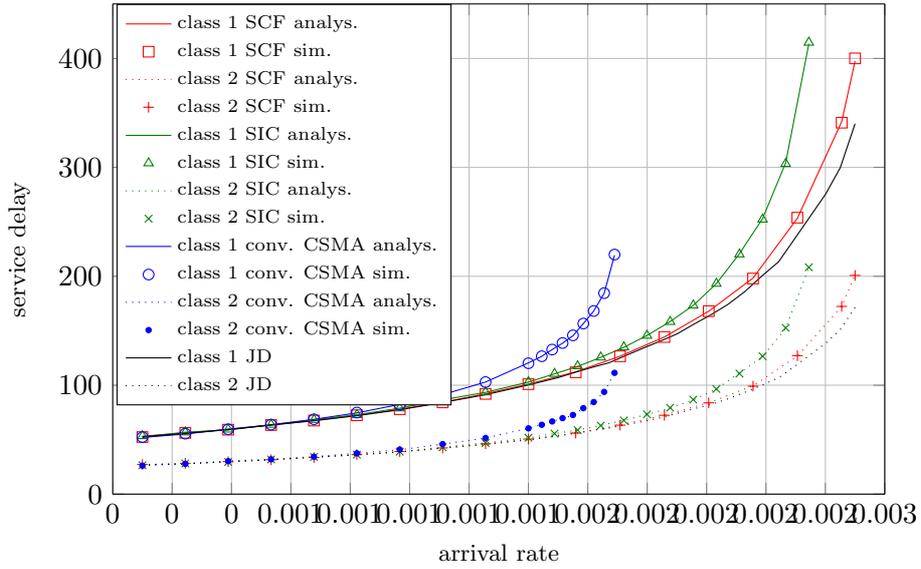} 
\caption{Service delay performance of inhomogeneous persistent CSMA schemes for a $2$-class network at \SNR~ =~ $6\,$dB with message rate $= 1$, $N_1=20$, $N_2 = 10$, $p_1 = 1/40$, $p_2= 1/20$, $K = 1$, $\kappa=10$.}
\label{fig:service_delay_two_class}
\end{figure}

\begin{figure}[t]
\centering 
\makeatletter
\if@twocolumn 
\setlength\figureheight{5.5cm} 
\setlength\figurewidth{0.83\columnwidth} 
\else
\setlength\figureheight{6.5cm} 
\setlength\figurewidth{.65\columnwidth} 
\fi
\makeatother 
\input{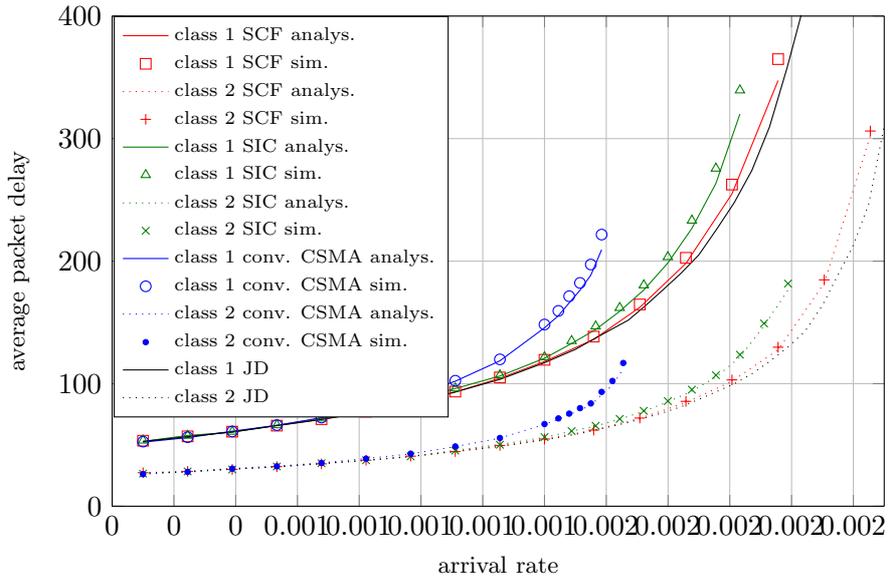} 
\caption{Average total delay (packet delay) performance of inhomogeneous persistent CSMA schemes for a $2$-class network.  The simulation setup is the same as that of Fig. \ref{fig:service_delay_two_class}.}
\label{fig:total_delay_two_class}
\end{figure}

We next look at the service delay and packet delay of symmetric MPR techniques in low and high $\SNR$ regimes. Figures~\ref{fig:service_delay_two_class},~\ref{fig:total_delay_two_class}\, demonstrates the service and packet delays as the arrival rate changes in a $2$-class network with $N=30$ ($N_1=20$ users in class-1) at $\SNR~=6$dB. We see that the analytical delay very well agrees with the simulation. Note that, in contrast to the existing work, the delay performance is presented for each class separately. In addition, it can be seen that CSMA with SCF is stable for a larger range of arrival rates. Similarly, the same behavior in terms of better performance of CSMA with SCF can be seen in Figures \ref{fig:service_delay_two_class_SNR15},~\ref{fig:total_delay_two_class_SNR15}\, for the same network at $\SNR~=15$dB.

\begin{figure}[t]
\centering 
\makeatletter
\if@twocolumn 
\setlength\figureheight{5.5cm} 
\setlength\figurewidth{0.83\columnwidth} 
\else
\setlength\figureheight{6.5cm} 
\setlength\figurewidth{.65\columnwidth} 
\fi
\makeatother
\input{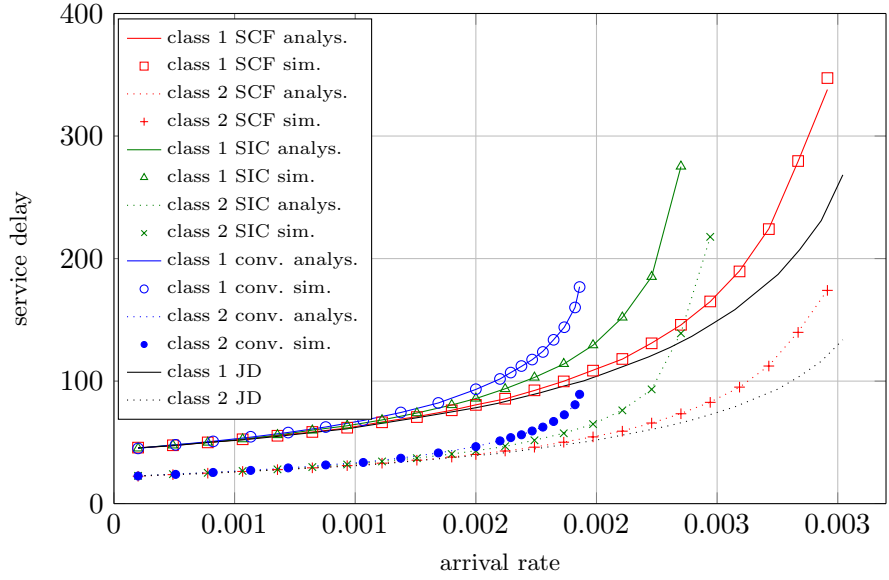} 
\caption{Service delay performance of inhomogeneous persistent CSMA schemes for a $2$-class network at \SNR~ =~ $15\,$dB with message rate $= 2$, $N_1=20$, $N_2 = 10$, $p_1 = 1/40$, $p_2= 1/20$, $K = 1$, $\kappa=10$.}
\label{fig:service_delay_two_class_SNR15}
\end{figure}

\begin{figure}[t]
\centering 
\makeatletter
\if@twocolumn 
\setlength\figureheight{5.5cm} 
\setlength\figurewidth{0.83\columnwidth} 
\else
\setlength\figureheight{6.5cm} 
\setlength\figurewidth{.65\columnwidth} 
\fi
\makeatother 
\input{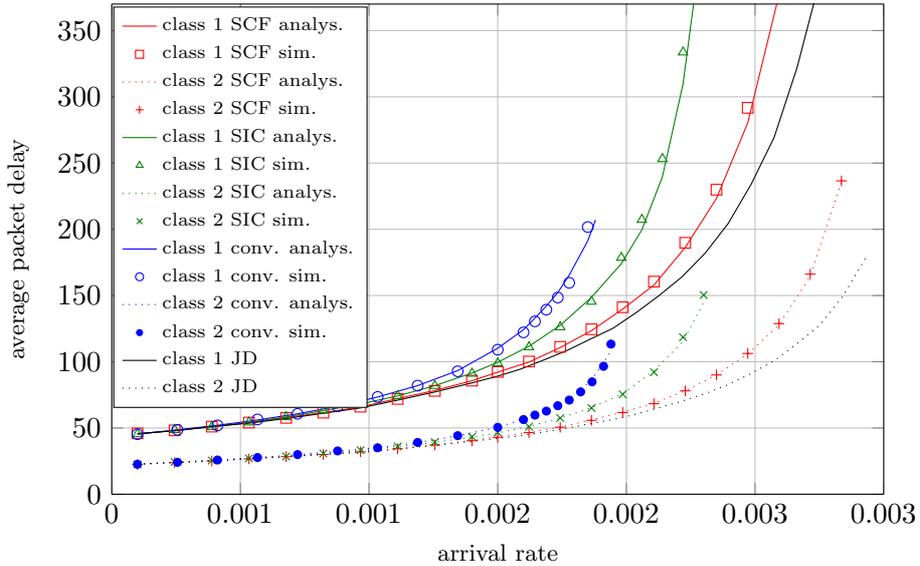} 
\caption{Average total delay (packet delay) performance of inhomogeneous persistent CSMA schemes for a $2$-class network.  The simulation setup is the same as that of Fig. \ref{fig:service_delay_two_class_SNR15}.}
\label{fig:total_delay_two_class_SNR15}
\end{figure}

\section{Conclusions}

In this work, we have studied inhomogeneous CSMA with symmetric MPR. In particular, we have derived throughput and delay expressions, which are asymptotically exact as the number of users grows. Based on these expressions, we have provided some theoretical guidelines for the network design, and evaluated the performance of various MPR techniques in terms of throughput, packet delay and service delay. Our work not only makes some progress on an open problem, but also sheds some light on the system design, highlighting some interesting properties of C\&F-based CSMA.

\begin{appendices}
\vspace{-.3cm}
\section{Proof of Lemma \ref{lem::unimodal}} \label{appx::unimodal}
Define $h(\gamma) \triangleq \gamma\chi(\gamma)e^{-\gamma}$. Let $\tilde\gamma$ be the maximizer of $h$. We assume that $h$ is unimodal\footnote{It turns out that the same condition for unimodality of $h$ suffices for $f$ to be unimodal. Note that when $\tau=1$ persistent CSMA system reduces to slotted ALOHA system where we need to impose unimodality to be able to characterize an approximate stability region.}. Note that $f(\gamma) = \frac{h(\gamma)}{e^{-\gamma}+\tau(1-e^{-\gamma})}$ is a continuous function over $[0, \gamma_0]$ as its denominator is always positive. In order for $f$ to be unimodal it should have exactly one maximum, thus its first order derivative should have only one positive root. Setting the first order derivative of $f$ equal to zero we get
\makeatletter
\if@twocolumn
\begin{align} \label{proof:unimodal:der}
\nonumber
&\tau\left( \chi(\gamma) + \gamma \frac{\partial \chi(\gamma)}{\partial\gamma}
- \gamma\chi(\gamma) \right) e^{\gamma} \\
&\qquad- (\tau-1)\left(\chi(\gamma) + \gamma\frac{\partial\chi(\gamma)}{\partial \gamma}  \right) = 0
\end{align}
\else
\begin{equation} \label{proof:unimodal:der}
\tau\left( \chi(\gamma) + \gamma \frac{\partial \chi(\gamma)}{\partial\gamma} - \gamma\chi(\gamma) \right) e^{\gamma} - (\tau-1)\left(\chi(\gamma) + \gamma\frac{\partial\chi(\gamma)}{\partial \gamma}  \right) = 0 \,.
\end{equation}
\fi
\makeatother
Equivalently
\[
g(\gamma) \triangleq \dot h(\gamma) - (\tau-1) f(\gamma)e^{-\gamma}= 0 \,,
\]
where $\dot h(\gamma)$ is the first order derivative of $h$. Observe that $g(0) > 0$ and since $\dot h(\gamma_0)<0$ for any $\gamma_0 > \tilde \gamma$ (as $h$ is unimodal) then $g(\gamma_0) < 0$. Hence $g$ has at least one positive root in $[0, \gamma_0]$.
By substituting Taylor series expansion of exponential function in (\ref{proof:unimodal:der}) we have
\begin{align} \label{derv_poly}
p(\gamma) \triangleq \tau\sum_{i=0}^{M} a_i \gamma^i \sum_{n\ge0} \frac{\gamma^n}{n!} - (\tau-1)\sum_{i=0}^{M-1}b_i \gamma^{i} = 0 \,,
\end{align}
where 
\begin{align}\label{eq::coeffs}
a_i = \frac{i+1}{i!} q_{i+1} - \frac{1}{(i-1)!}q_i \quad \text{and} \quad b_i = \frac{i+1}{i!} q_{i+1} \,.
\end{align}
Equation (\ref{derv_poly}) is a degree $\infty$ polynomial in $\gamma$. We showed that $p(\gamma)$ has at least one positive root. To prove that (\ref{derv_poly}) has exactly one positive root we apply generalized Descartes' rule of signs \cite{Curtiss:descartes:1918} to impose one sign difference to the coefficients of $p(\gamma)$.

Let $c_k$ represent the coefficients of $p(\gamma)$. Then
\begin{align}\label{proof:coeffs}
c_k =
\begin{cases}
\tau\sum_{n=0}^{k} \frac{a_{n}}{(k-n)!} - (\tau-1)b_k  \qquad  &0 \le k \le M-1 \\
\tau\sum_{n=0}^{M} \frac{a_{n}}{(k-n)!} \qquad &\quad\,\,\,\,\, k>M-1
\end{cases}
\end{align}
We observe from (\ref{eq::coeffs}) that $a_M < 0$ as $q_{M+1} = 0$. Suppose that $a_i \ge 0$ for all $0 \le i < M$. In this case $c_k  \ge 0$ for all  $0 \le i < M$. This can be shown by induction on $k$. For $k=0$, clearly $c_0 > 0$. Suppose that for some $k-1 > 0$, $c_{k-1} >0$. Then
\makeatletter
\if@twocolumn
\begin{align*}
c_{k} &= \tau \sum_{n=0}^{k} \frac{a_{n}}{(k-n)!} - (\tau-1)b_k \\
&= c_{k-1} + \tau\frac{(k+1)q_{k+1}-kq_k}{k!}
\end{align*}
\begin{align*}
&\quad- (\tau-1)\frac{k+1}{k!}q_{k+1}+ (\tau-1)\frac{k}{(k-1)!}q_k
\end{align*}
\begin{align*}
&= c_{k-1} + \tau\frac{(k+1)q_{k+1}-kq_k}{k!} \\
&\qquad+ (\tau-1)\frac{k^2 q_k-(k+1)q_{k+1}}{k!} \\
&> c_{k-1} + \tau\frac{(k+1)q_{k+1}-kq_k}{k!}
\end{align*}
\begin{align*}
&\qquad+ (\tau-1)\frac{k q_k-(k+1)q_{k+1}}{k!} \\
&=c_{k-1} + \frac{(k+1)q_{k+1}-kq_k}{k!}\\
& > 0
\end{align*}
\else
\begin{align}
c_{k} &= \tau \sum_{n=0}^{k} \frac{a_{n}}{(k-n)!} - (\tau-1)b_k \\
&= c_{k-1} + \tau\frac{(k+1)q_{k+1}-kq_k}{k!} - (\tau-1)\frac{k+1}{k!}q_{k+1}+ (\tau-1)\frac{k}{(k-1)!}q_k \\
&= c_{k-1} + \tau\frac{(k+1)q_{k+1}-kq_k}{k!} + (\tau-1)\frac{k^2 q_k-(k+1)q_{k+1}}{k!} \\
&> c_{k-1} + \tau\frac{(k+1)q_{k+1}-kq_k}{k!} + (\tau-1)\frac{k q_k-(k+1)q_{k+1}}{k!} \\
&=c_{k-1} + \frac{(k+1)q_{k+1}-kq_k}{k!}\\
& > 0
\end{align}
\fi
\makeatother
The last step follows if $(k+1)q_{k+1}-kq_k \geq 0$ or equivalently $a_k \geq 0$.
 Now if the rest of the coefficients be negative then according to the Descartes' rule of sign the number of positive roots of a polynomial with real coefficients ordered by descending variable exponent is either equal to the number of sign differences between consecutive nonzero coefficients, or is less than it by an even number. Therefore in order to have one positive root, it is sufficient to have one sign difference. Equivalently, all the coefficients $c_k$ for $k\ge M$ should be negative. We observe that the assumption $a_i \ge 0$ ensures that $c_k < 0$ for all $k \ge M_0$ for some $M_0$. Because
\makeatletter
\if@twocolumn
\begin{align*}
c_k &= \tau\sum_{n=0}^{M} \frac{a_{n}}{(k-n)!} \\
&= \tau \left(\frac{a_0}{k!} + \frac{a_1}{(k-1)!} + \cdots + \frac{a_M}{(k-M)!} \right) \\
&= \tau \bigg(\frac{q_1}{k!0!} + \frac{2q_2-q_1}{1!(k-1)!} + \cdots \\
&\qquad\quad + \frac{(M+1)q_{M+1}-Mq_M}{M!(k-M)!} \bigg)
\end{align*}
\begin{align*}
& = \frac{\tau}{k!} \bigg( \binom{k}{0}q_1 + \binom{k}{1}(2q_2-q_1) + \cdots \\
&\qquad\quad+ \binom{k}{M}((M+1)q_{M+1}-Mq_M)   \bigg)
\end{align*}
\begin{align*}
&\le \frac{\tau}{k!} \bigg( \binom{k}{M}q_1 + \binom{k}{M}(2q_2-q_1) + \cdots \\
&\qquad\quad+ \binom{k}{M}((M+1)q_{M+1}-Mq_M)   \bigg)\\
&<0\,.
\end{align*}
\else
\begin{align*}
c_k &= \tau\sum_{n=0}^{M} \frac{a_{n}}{(k-n)!} \\
&= \tau \left(\frac{a_0}{k!} + \frac{a_1}{(k-1)!} + \cdots + \frac{a_M}{(k-M)!} \right) \\
&= \tau \left(\frac{q_1}{k!0!} + \frac{2q_2-q_1}{1!(k-1)!} + \cdots + \frac{(M+1)q_{M+1}-Mq_M}{M!(k-M)!} \right)\\
& = \frac{\tau}{k!} \left( \binom{k}{0}q_1 + \binom{k}{1}(2q_2-q_1) + \cdots + \binom{k}{M}((M+1)q_{M+1}-Mq_M)   \right) \\
&\le \frac{\tau}{k!} \left( \binom{k}{M}q_1 + \binom{k}{M}(2q_2-q_1) + \cdots + \binom{k}{M}((M+1)q_{M+1}-Mq_M)   \right)\\
&<0\,.
\end{align*}
And $c_k > 0$ for $M \leq k \leq M_0$ which follows by induction starting from $k=M$ and $c_M> 0$.
\fi
\makeatother
 So if $a_i \geq 0$ then $f$ is unimodal. Equivalently we have to have $((i+1)q_{i+1} - iq_{i})/i! \geq 0$ for all $1 \leq i \leq M-1$ which yields the result.

\vspace{-.3cm}
\section{Proof of Theorem \ref{thm::delay::CSMA} } \label{appendix::CSMA::delay}

Suppose that there are $n$ packets in the buffer of a class-$v$ queue upon the arrival of a new packet. An arrival can occur at an idle, collision or a success time slot. So we compute the delay conditioned on the state of the system. Let $T_v(n)$ be the delay being experienced by a newly arrived packet at a class-$v$ queue until it is successfully transmitted.  Conditioned on the state of a time slot we have the followings cases: 

\begin{itemize}
\item \emph{Idle}: At the time of a new arrival, there is a chance that one of the $n$ packets is already in the process of a successful transmission. Therefore, the delay being experienced by the newly arrived packet can be written as follows
\begin{itemize}
\makeatletter
\if@twocolumn
\item if $n=0$ then  $T_v(n) = \D$
\item if $n>0$ then $T_v(n) =    n\D$ with probability
$$\frac{1/\tau \Psucc}{\Pidle + 1/\tau (1-\Pidle)}$$
and $T(n) =  (n+1)\D$ otherwise \,.
\else
\item if $n=0$ then  $T_v(n) = \D$
\item if $n>0$
\[
T_v(n) =  \left\{ 
  \begin{array}{l l l }
    n\D            &    \quad     \text{$\frac{1/\tau \Psucc}{\Pidle + 1/\tau (1-\Pidle)}$}\\
    (n+1)\D      &    \quad     \text{otherwise} \,.
  \end{array} \right.
\]
\fi
\makeatother
\end{itemize}
\item \emph{Collision or Success}: If the new arrival occurs in a transmission period (whether a collision or successful transmission), it is either at the first slot of a transmission period or any other time slot.
\begin{itemize}
\item with probability $1/\tau (\Pcol+\Psucc)$ the arrival occurs in the first time slot and $T(n)$ can be computed as we discussed above for the idle case.
\item with probability $(\tau-1)/\tau (\Pcol+\Psucc)$ the arrival occurs in any time slot but the first time slot. For all $n\geq0$ we then have
\[
T_v(n) =  R+(n+1)\D
\]
where $R$ is a random variable representing the residual time until a transmission period ends.
\end{itemize}
\end{itemize}
Let $Q_{v,n}$ be the probability that a class-$v$ user has $n$ packets in its buffer at the steady state. Further let $\T = \sum_{n \geq 0} T_v(n) Q_{v,n} + o(1)$ be the total delay of a class-$v$ packet. We write (\ref{proof::total::delay}).
\makeatletter
\if@twocolumn
\begin{floatEq}
\begin{align}
\label{proof::total::delay}
\T &= \sum_{n>0} \Big ( \Pidle +  1/\tau \big(\Pcol+\Psucc \big) \Big ) \\ \nonumber
 &\qquad \bigg [  n\D \bigg ( \frac{1/\tau \Psucc}{\Pidle + 1/\tau (1-\Pidle)} \bigg ) \\ \nonumber
&+  (n+1)\D \bigg ( 1 - \frac{1/\tau \Psucc}{\Pidle + 1/\tau (1-\Pidle)} \bigg )   \bigg ] \\ \nonumber
&\quad + Q_{(v,0)} \Big ( \D \Pidle + \D 1/\tau \big(\Pcol+\Psucc \big) \Big ) \\ \nonumber
&\qquad + \frac{\tau-1}{\tau}(\Pcol+\Psucc) \sum_{n \geq 0}(R+(n+1)\D) Q_{(v,n)} + o(1)\,.
\end{align}
 \end{floatEq}
 \else
 \begin{align}
\label{proof::total::delay}
\T &= \sum_{n>0} \Big ( \Pidle +  1/\tau \big(\Pcol+\Psucc \big) \Big ) \\ \nonumber
 &\qquad \bigg [  n\D \bigg ( \frac{1/\tau \Psucc}{\Pidle + 1/\tau (1-\Pidle)} \bigg ) 
 \end{align}
 \begin{align}
 \nonumber
&+  (n+1)\D \bigg ( 1 - \frac{1/\tau \Psucc}{\Pidle + 1/\tau (1-\Pidle)} \bigg )   \bigg ] \\ \nonumber
&\quad + Q_{v,0} \Big ( \D \Pidle + \D 1/\tau \big(\Pcol+\Psucc \big) \Big )
\end{align}
\begin{align}
\nonumber
&\qquad + \frac{\tau-1}{\tau}(\Pcol+\Psucc) \sum_{n \geq 0}(R+(n+1)\D) Q_{v,n} + o(1)\,.
\end{align}
\fi
\makeatother
After some manipulation of (\ref{proof::total::delay}) and since $(1 - Q_{v,0})\Psucc = \lambda_v$ we obtain the following
\begin{align}
\T = \frac{\D (1 - \lambda_v) + d_v}{1 - \lambda_v \D} + o(1)
\end{align}
where 
\[
d_v =  (\tau-1) \bigg( \frac{1}{2}\left(1- \Pidle\right) + \frac{1}{\tau}\lambda_v \D \bigg)
\]
which results in Theorem \ref{thm::delay::CSMA}.

\end{appendices}

\bibliographystyle{IEEEtran}
\bibliography{refuw}
\end{document}